\documentclass[11pt]{article}

\usepackage{fullpage}
\usepackage{complexity}
\usepackage{cite}

\title{Protecting the Connectivity of a Graph Under Non-Uniform Edge Failures}
\author{
Felix Hommelsheim\thanks{\texttt{\{fhommels,nicole.megow\}@uni-bremen.de}, University of Bremen, Germany.} 
\and 
Zhenwei Liu$^*$\thanks{\texttt{\{lzw98,zgc\}@zju.edu.cn}, Zhejiang University, China. Supported in part by NSFC (No. 12271477).} \and 
Nicole Megow\footnotemark[1] \and 
Guochuan Zhang\footnotemark[2]
}

\usepackage{tikz}
\usetikzlibrary{positioning}
\usetikzlibrary{decorations.pathreplacing,angles,quotes}
\usepackage{mathtools}
\usepackage{multicol}
\usepackage{amsmath}
\usepackage{amssymb}
\usepackage{amsthm}
\usepackage{thmtools}
\usepackage{subfig}
\usepackage{graphicx}
\usepackage{url}
\usepackage{hyperref}
\usepackage{cleveref}

\newtheorem{theorem}{Theorem}
\newtheorem{remark}[theorem]{Remark}

\newtheorem{lemma}[theorem]{Lemma}

\newtheorem{definition}{Definition}

\newcommand{\bigO}{\mathcal{O}}

\newcommand{\opt}{{\normalfont\textsc{Opt}}}

\newcommand{\steinerc}[2]{$(#1,#2)$-Steiner-Connectivity Preservation}
\newcommand{\steinercshort}[2]{$(#1,#2)$-SCP}
\newcommand{\globalc}[2]{$(#1,#2)$-Global-Connectivity Preservation}
\newcommand{\globalcshort}[2]{$(#1,#2)$-GCP}
\newcommand{\stc}[2]{$(#1,#2)$-$s$-$t$-Connectivity Preservation}
\newcommand{\stcshort}[2]{$(#1,#2)$-stCP}

\newcommand{\fgcshort}[2]{$(#1,#2)$-GFND}
\newcommand{\fnd}[2]{$(#1,#2)$-Flexible Network Design}
\newcommand{\fndshort}[2]{$(#1,#2)$-FND}
\newcommand{\fndst}[2]{$(#1,#2)$-$s$-$t$-Flexible Network Design}
\newcommand{\fndstshort}[2]{$(#1,#2)$-stFND}

\usepackage[colorinlistoftodos,prependcaption,textsize=scriptsize]{todonotes}

\begin{document}

\maketitle

\begin{abstract}
We study the problem of guaranteeing the connectivity of a given graph by protecting or strengthening edges. Herein, a protected edge is assumed to be robust and will not fail, which features a non-uniform failure model. We introduce the \steinerc{p}{q} problem where we protect a minimum-cost set of edges such that the underlying graph maintains $p$-edge-connectivity between given terminal pairs against edge failures, assuming at most $q$ unprotected edges can fail. We design polynomial-time exact algorithms for the cases where $p$ and $q$ are small and approximation algorithms for general values of $p$ and $q$.
Additionally, we show that when both $p$ and $q$ are part of the input, even deciding whether a given solution is feasible is \NP-complete. 
This hardness also carries over to Flexible Network Design, a research direction that has gained significant attention. 
In particular, previous work focuses on problem settings where either $p$ or $q$ is constant, for which our new hardness result now provides justification.
\end{abstract}

% \clearpage
%\setcounter{page}{1}
\section{Introduction}
\label{sec:intro}

In today's interconnected world, the robustness of infrastructures is crucial, particularly in the face of potential disruptions. 
This applies to road networks, communication grids, and electrical systems alike, where the ability to maintain functionality despite failures is paramount.
Resilience in critical infrastructures requires the incorporation of redundancy to withstand unforeseen challenges.
Survivable Network Design (SND) addresses the fundamental %imperative 
need of ensuring not just connectivity but robust connectivity. It goes beyond the conventional concept of linking two entities by recognizing  the need for multiple, resilient~connections.

Beyond its practical applications, SND is a fundamental problem in combinatorial optimization and approximation algorithms.
In its classical setting, we are given a graph $G=(V,E)$ with edge costs $c: E \rightarrow \mathbb{R}$ and connectivity requirements $k(s,t)$ for each vertex pair $s,t \in V$. 
The goal is to find a minimum-cost subset of edges $X \subseteq E$ such that $H=(V,X)$ is $k(s,t)$-connected for each $s,t \in V$, i.e., in $H$ there are $k(s, t)$ many edge-disjoint paths for each vertex pair $s,t \in V$.
%In a different view, 
This means that $s$ and $t$ are still connected in $H$ after the deletion of any $k(s,t) -1$ edges of $H$.
SND is \APX-hard and the current best approximation factor of~$2$ is achieved by Jain's famous iterative rounding algorithm~\cite{DBLP:journals/combinatorica/Jain01}.
It is a long-standing open question whether this factor $2$ can be improved, even for many special cases of SND.

Instead of building a new network from scratch, many real-world applications as well as the research community consider augmentation problems, in which we are already given some network, and the task is to robustify the network to withstand failures. The most common model is to increase the connectivity of a given network by adding more links \cite{DBLP:journals/siamcomp/EswaranT76,DBLP:journals/siamcomp/FredericksonJ81}.
However, adding new links to a real-world network may be costly or even impractical.

This motivates the investigation of the problem of increasing the connectivity of a network by {\em protecting} or strengthening edges, as %already 
has been proposed by Abbas et al.~\cite{abbas2017improving} in a %more 
practical context.
In this paper, we formally introduce the problem \emph{\steinerc{p}{q} (\steinercshort{p}{q})}:
Given an undirected graph $G= (V, E)$, possibly with parallel edges, nonnegative costs $c : E \rightarrow \mathbb{R}_+$ and $k$ terminal pairs $(s_i, t_i) \in V \times V$. %Parallel edges are allowed.
The task is to %protect 
identify a minimum-cost set of edges $X \subseteq E$ such that for any edge set $F \subseteq E \setminus X$ with $|F| \leq q$, there are $p$ edge-disjoint paths between any terminal pair $(s_i, t_i)$ in $(V, E \setminus F)$.
In other words, the task is to protect a minimum-cost subset of the edges $X \subseteq E$ such that, no matter which~$q$ unprotected edges from $E \setminus X$ are removed from $G$, there are still~$p$ edge-disjoint paths between any terminal pair.
%If there is only a single terminal pair $(s, t)$, we call the problem \emph{\stc{p}{q} (\stcshort{p}{q})}.
%The special case in which each pair of vertices is a terminal pair is called \emph{\globalc{p}{q} (\globalcshort{p}{q})}.
We refer to the special case with a single terminal pair $(s, t)$ by \emph{\stc{p}{q} (\stcshort{p}{q})} and the other extreme case in which \emph{each} pair of vertices is a terminal pair as \emph{\globalc{p}{q} (\globalcshort{p}{q})}. Using this notion, Abbas et al.~\cite{abbas2017improving} %, who introduced this problem, only 
considered \globalcshort{1}{q} and proposed to start from a minimum spanning tree and remove unnecessary edges, which does not admit bounded approximation factors.
Zhang et al.~\cite{zhang2017enhancing} used mixed-integer linear programming to solve \globalcshort{1}{q} and \steinercshort{1}{q}.
 Bienstock and Diaz~\cite{DBLP:journals/siamcomp/BienstockD93} considered a special case of \steinercshort{1}{q}, the all-pair connectivity among a set of terminals, and showed a polynomial-time algorithm for $q \leq 2$ and the \NP-hardness for $q=8$. 

The distinction between protected and unprotected edges has a similar flavor as the \fnd{p}{q} (\fndshort{p}{q})
%Flexible Network Design} 
problem~\cite{DBLP:conf/swat/AdjiashviliHMS22,DBLP:journals/mp/AdjiashviliHM22, DBLP:journals/mp/BoydCHI24,bansal2023improved,DBLP:journals/corr/abs-2308-15714,chekuri2023approximation,DBLP:conf/approx/BansalCGI23,bansal2024improvedapproximationalgorithmsflexible}.
The input is an edge-weighted undirected graph together with a set of terminal pairs, where the edges are either safe or unsafe. 
The goal is to find a minimum-cost subgraph such that any terminal pair remains $p$-edge-connected after the failure of any $q$ \emph{unsafe} edges. 
Also here different versions have been studied, e.g., \fgcshort{p}{q}, the global connectivity version (each pair of vertices is a terminal pair) or \fndstshort{p}{q}, the $s$-$t$ version (only one terminal pair).
In contrast to their model, in our setting we strengthen \emph{existing} edges rather than building a network from scratch. 
For $p=1$, \steinercshort{1}{q} reduces to \fndshort{1}{q}.

Given an instance of \steinercshort{1}{q}, we construct an instance of \fndshort{1}{q} as follows.
We use the same underlying graph but replace each edge by two parallel edges: one unsafe edge of cost zero and one safe edge with cost equal to the cost of protecting the edge in the instance of \steinercshort{1}{q}.
Since the unsafe edges have cost zero, we can assume that any feasible solution includes all unsafe edges.
Then, a feasible solution to \fndshort{1}{q} can be transformed into a feasible solution to \steinercshort{1}{q} with the same cost by protecting the selected safe edges and vice versa.
For $p > 1$, however, we are not aware of any reduction from \steinercshort{p}{q} to \fndshort{p}{q}.

\subsection{Our Contribution}
In this paper, we study Connectivity Preservation problems %the problem \steinerc{p}{q} 
in terms of algorithms, complexity, and approximability.

The \steinerc{1}{q} problem is \APX-hard if $q$ is part of the input: When $q$ is sufficiently large, say, larger than $|E|$, any feasible solution to \steinercshort{1}{q} includes at least one edge in any 
terminal-separating cut (precise definitions are given in \Cref{sec:preliminaries}).
%cut that separates some terminal pair. 
Hence, it is equivalent to Steiner Forest, which is \APX-hard~\cite{DBLP:journals/ipl/BernP89}. 
Similarly, \globalcshort{p}{q} is \APX-hard for any $p \geq 2$, as it contains the minimum $p$-edge-connected spanning subgraph problem~\cite{DBLP:conf/waoa/Pritchard10} as a special case when $q$ is sufficiently large.
We strengthen this by showing that \globalcshort{1}{q} is also \NP-hard when $q$ is part of the input.

Motivated by the problem hardness, we first study cases when $p$ or $q$ is small. 
We obtain polynomial-time exact algorithms for \steinercshort{p}{1} for any $p \geq 1$ and \steinercshort{1}{2} as well as a polynomial-time exact algorithm for \globalcshort{2}{2}.
We emphasize that \steinercshort{p}{q} generalizes \globalcshort{p}{q} and hence any algorithm for \steinercshort{p}{q} also works for \globalcshort{p}{q}.

\begin{theorem}[summarized]
\label{thm:main:polytime}
    There are polynomial-time exact algorithms for \steinerc{p}{1} for any $p\geq 1$ and \steinerc{1}{2}. 
    Furthermore, there is a polynomial-time exact algorithm for \globalc{2}{2}. 
\end{theorem}

The first result for \steinercshort{p}{1} is quite easily obtained by observing that a solution is only feasible if every edge contained in some terminal-separating cut of size at most $p$ is protected.

% The polynomial-time algorithm for \steinercshort{1}{2} is more involved and relies on a decomposition of terminal-separating cuts of size 2.
% We first reduce the problem to the case in which $G$ is assumed to be $2$-edge-connected. 
% Then, it remains to protect one of the edges in each terminal-separating cut of size $2$.
% To do so, we compute certain cycles through 2-edge-connected components of the graph, with the property that each pair of edges of the cycle is a $2$-cut (not necessarily a terminal-separating cut).
% We show that an optimum solution to the original problem can be obtained by combining optimum solutions of certain subproblems: $(i)$~We have one instance for each $2$-edge-connected component, by adding certain pseudo-edges that represent the connectivity from the cycle. $(ii)$ The problem on the cycle, i.e., the problem when contracting each of the $2$-edge-connected components to a single vertex. 

The polynomial-time algorithm for \steinercshort{1}{2} is more involved and relies on a decomposition of terminal-separating cuts of size $2$. We reduce the problem to the case in which $G$ is assumed to be $2$-edge-connected. 
Then, it remains to protect one edge in each terminal-separating cut of size $2$.
To do so, we decompose the problem into subproblems that consist either of a $2$-edge-connected component or a cycle that can be solved independently.

The polynomial-time algorithm for \globalcshort{2}{2} is the most involved exact algorithm. 
We first show that we can assume without loss of generality that $G$ is $3$-edge-connected,
%, as otherwise for a given 2-edge-cut we can decompose the problem into two subproblems.
%Given a $3$-edge-connected graph, the remaining task is to
which reduces the problem to selecting a minimum-cost set of edges containing at least $2$ edges from each $3$-edge-cut. 
Equivalently, we select a maximum-weight set of edges such that we pick at most $1$ edge from each $3$-edge-cut.
Using the well-known tree-representation of minimum cuts~\cite{dinitz1973structure}, we model this problem as a multi-commodity flow problem on a tree: given a tree and a set of weighted paths on the tree, find a maximum weighted subset of paths that are pairwise edge-disjoint, which was first introduced in~\cite{DBLP:journals/algorithmica/GargVY97} for the unweighted case.
We solve the weighted problem via dynamic programming, which might be of independent interest.

 We complement our %polynomial-time 
 exact algorithms for small $p$ and $q$ with hardness and approximation results for more general cases. For \steinercshort{p}{q}, if both $p$ and $q$ are part of the input, we show that there is no polynomial-time algorithm verifying the feasibility of any given solution, unless \P$=$\NP, even in the case of $s$-$t$-connectivity.
This rules out an $\alpha$-approximation algorithm for any $\alpha$. 
Our technique is based on a reduction from $k$-Clique on $d$-regular graphs.
Note that the corresponding solution verifying problems of \stcshort{p}{q} and \fndstshort{p}{q} are essentially the same: given sets of protected (resp. safe) and unprotected (resp. unsafe) edges, decide whether there is an $s$-$t$-cut that has no more than $p+q-1$ edges and has less than $p$ protected (resp. safe) edges. 
Hence, the hardness of verifying a solution also carries over to \fndst{p}{q} and  justifies why previous work on this problem only discusses the cases where either $p$ or $q$ is constant~\cite{DBLP:conf/swat/AdjiashviliHMS22,chekuri2023approximation}.

% \begin{restatable}{theorem}{stcuthard}\label{thm:stcut_hard}
% When both $p$ and $q$ are part of the input, verifying the feasibility of any solution to \stc{p}{q} or \fndst{p}{q} is \NP-complete. 
% %even in perfect graphs. 
% As a consequence, there is no $\alpha$-approximation for these problems for any $\alpha$ unless $\P=\NP$.
% \end{restatable}
\begin{restatable}{theorem}{stcuthard}\label{thm:stcut_hard}
When both $p$ and $q$ are part of the input, verifying the feasibility of a solution to \stc{p}{q} or \fndst{p}{q} is \NP-complete, 
even in perfect graphs. 
Hence, they do not admit an $\alpha$-approximation for any $\alpha$ unless~$\P=\NP$.
% Hence, there is no $\alpha$-approximation for any $\alpha$ unless $\P=\NP$.
\end{restatable}

 On the algorithmic side, if $q$ is a constant, we can enumerate all "bad" edge sets whose removal destroy the $p$-edge-connectivity.
 Since any feasible solution intersects with or hits these "bad" sets, it reduces to the hitting set problem and admits a $q$-approximation.
 Then we focus on the case where $p$ is a constant.
 We first give a $q$-approximation for \steinercshort{1}{q} and extend it to \steinercshort{p}{q} based on a primal-dual framework \cite{DBLP:journals/combinatorica/WilliamsonGMV95, DBLP:conf/soda/GoemansGPSTW94}.
 In the framework, we iteratively protect one more edge in each \emph{critical} cut (precise definition follows), which is very similar to the problem \steinercshort{1}{q}.
 We obtain the following result.

\begin{restatable}{theorem}{thmapprox} \label{thm:steiner-approx}
    There is a polynomial-time $\bigO((p+q) \log p)$-approximation algorithm for \steinerc{p}{q} when $p$ is a constant. 
\end{restatable}

The global connectivity problem \globalcshort{p}{q} has more symmetric structure, which enables us to remove the requirement of $p$ being constant. 
Hence, the above result directly carries over to this case without the assumption on $p$.
In addition, we design a different set-cover based augmentation process.
This algorithm relies on the fact that there is only a polynomial number of critical cuts to be augmented, which is not true for \steinercshort{p}{q}. Combining the two algorithms, we obtain the following result.

\begin{restatable}{theorem}{thmapproxtwo}
\label{thm:global-apx}
There is a polynomial-time $\bigO( \log p \cdot \min\{p+q, \log n\})$-approximation algorithm for \globalc{p}{q}.
\end{restatable} 

We  obtain further results for special cases by showing reductions to known problems. 
Since the {\em Augmenting Small Cuts} problem~\cite{bansal2023improved} generalizes \globalc{1}{q}, we obtain a 5-approximation building on %results from~
\cite{nutov2024improved}.
Further, we show that \stc{1}{q} is equivalent to the {\em Minimum Shared Edge} problem; formally defined in \Cref{{sec:apx-algos}}.
This reduction implies, due to earlier work, a fixed-parameter tractable (FPT) algorithm parameterized by $q$ if the graph is undirected~\cite{DBLP:journals/jcss/FluschnikKNS19} and an \XP-algorithm (slice-wise polynomial) for directed graphs~\cite{DBLP:journals/algorithmica/AssadiENYZ14}. 
Further, for directed graphs the equivalence of the two problems implies a strong inapproximability bound of $\Omega(2^{\log^{1-\epsilon}\max\{q,n\}})$, unless $\NP \subseteq \DTIME(n^{polylog(n)})$~\cite{DBLP:journals/jco/OmranSZ13}.
Since \stc{1}{q} is a special case of \fndst{1}{q}, namely, %i.e., the special case of 
\fnd{1}{q} with only a single terminal pair, this strong hardness bound also holds for \fndstshort{1}{q}, where the best-known lower bound on the approximation ratio is $\Omega(\log^{2-\varepsilon}q)$ unless $\NP \subseteq \mathsf{ZTIME}(n^{\polylog(n)})$~\cite{DBLP:conf/swat/AdjiashviliHMS22}.

\subsection{Related Work}
The \steinerc{p}{q} problem generalizes many well-known problems from survivable network design (SND), which itself
%\noindent \textbf{Survivable Network Design:} The survivable network design problem 
generalizes a collection of connectivity problems such as the minimum spanning tree problem, the Steiner tree and forest problem, or the minimum $k$-edge-connected spanning subgraph problem ($k$-ECSS). 

Many special cases of SND remain \APX-hard. 
This includes many  
augmentation problems, where typically the task is to increase the connectivity of a graph by at least $1$. 
If the set of edges to be added is unrestricted, the problem can be solved even in near-linear time~\cite{DBLP:journals/siamcomp/EswaranT76, DBLP:conf/stoc/CenLP22}, whereas the problem is \APX-hard 
otherwise~\cite{DBLP:journals/siamcomp/FredericksonJ81, DBLP:journals/siamcomp/KortsarzKL04}.
%Important augmentation problems 
Well-studied such problems include the Connectivity Augmentation problem~\cite{DBLP:conf/focs/TraubZ21, DBLP:conf/stoc/TraubZ23} and the Tree Augmentation problem~\cite{DBLP:conf/soda/TraubZ22, DBLP:conf/stoc/CecchettoTZ21}.

A problem of similar flavor was introduced by Adjiashvili, Stiller and Zenklusen \cite{DBLP:journals/mp/AdjiashviliSZ15}, who initiated the study of highly non-uniform failure models, called bulk-robustness. 
Therein,  a family of edge sets $\mathcal{F}$ 
is given
and the goal is to find a minimum-cost spanning subgraph $H$ such that $H \setminus F$ is connected for any $F \in \mathcal{F}$. 
They proposed an $\bigO(\log n + \log m)$-approximation algorithm for a generalized matroid setting. 
They also studied an $s$-$t$-connectivity variant and obtained an $\bigO(w^2\log n)$-approximation algorithm, where $w = \max_{F \in \mathcal{F}} |F|$. 
Recently, Chekuri and Jain \cite{chekuri2024approximation} considered the connectivity between multiple vertex pairs and achieved an $\bigO(w^2\log^2 n)$-approximation algorithm.

The aforementioned Flexible Network Design problem can be viewed as a problem between survivable network design and bulk-robustness, as it divides the edge set into safe and unsafe edges and only $q$ unsafe edges can fail simultaneously.
Since the work by Adjiashvili, Hommelsheim and M{\"{u}}hlenthaler~\cite{DBLP:journals/mp/AdjiashviliHM22} for $(1,1)$-GFND, there has been a lot of work on $(p,q)$-GFND. 
Most research focused on the case where either $p$ or $q$ is a small constant. 
Boyd et al.~\cite{DBLP:journals/mp/BoydCHI24} obtained $(q+1)$-approximation for $(1,q)$-GFND, a $4$-approximation for $(p,1)$-GFND and $\bigO(q \log n)$-approximation for $(p,q)$-GFND. 
Very recently, Bansal et al.~\cite{bansal2024improvedapproximationalgorithmsflexible} showed an improved $7$-approximation algorithm for $(1,q)$-GFND.
We refer to \cite{bansal2023improved,DBLP:journals/corr/abs-2308-15714,chekuri2023approximation, bansal2024improvedapproximationalgorithmsflexible} for a collection of results, including constant approximation for $(p,2),(p,3)$-GFND, $\bigO(q)$-approximation for $(2,q)$-GFND and $\bigO(p)$-approximation for $(p,4)$-GFND for even $p$. 
Parallel to $(p,q)$-GFND, Adjiashvili et al.~\cite{DBLP:conf/swat/AdjiashviliHMS22} considered the $s$-$t$-connectivity variant, 
to which some results in~\cite{chekuri2023approximation} translate.

    Another closely related problem is the Capacitated $k$-Connected Subgraph problem (Cap-$k$-ECSS)~\cite{DBLP:journals/algorithmica/ChakrabartyCKK15}.
    In this problem, we are given an undirected graph $G=(V, E)$ with edge costs $c : E \rightarrow \mathbb{R}_+$ and edge capacities $u : E \rightarrow \mathbb{Z}_+$.
    The goal is to find a minimum-cost spanning subgraph in which the capacity of any cut is at least $k$.
    Boyd et al.~\cite{DBLP:journals/mp/BoydCHI24} pointed out that $(1,q)$-GFND (hence also $(1,q)$-GCP) can be reduced to Cap-$(q+1)$-ECSS by setting the capacity of safe and unsafe edges to $q+1$ and $1$, respectively.
    For Cap-$k$-ECSS,
    the best-known approximation algorithms are $\bigO(\log n)$-approximation by
    Chakrabarty et al.~\cite{DBLP:journals/algorithmica/ChakrabartyCKK15} and $\bigO(\log k)$-approximation by Bansal et al.~\cite{bansal2024improvedapproximationalgorithmsflexible}.

\section{Preliminaries: Notation, Cut Formulation, Hardness}
\label{sec:preliminaries}

\textbf{Graph notation.\ }
For an undirected graph $G=(V, E)$ and a vertex set $S \subset V$, we use $\delta_G(S)$ to denote the set of edges with exactly one endpoint in $S$. 
We write $\delta(S)$ if the underlying graph is clear from the context. 
An edge cut $C$ is a subset of edges such that $G \setminus C$ has at least $2$ connected components. 
If $|C|=1$, we call $\{e \} = C$ a {\em bridge}. 
Further, if there is some terminal pair $(s_i,t_i)$ such that $s_i$ and $t_i$ are in different connected components in $G \setminus C$, we say $C$ is terminal-separating. 
Let $e=(u,v) \in E$. We use the notation of $G/e$ to denote the graph obtained from $G$ by contracting $e$, i.e., by deleting $e$ and identifying $u,v$. 
Let $G'=(V',E')$ be some subgraph of $G$. 
We slightly abuse the notation of contraction and use $G/G'$ to represent the graph obtained from $G$ by contracting all edges in $E'$.
Let $e \in E(G)$. We use $G - e$ to denote the graph $G \setminus \{e\}$. 
Similarly, we define $G+e$.

\textbf{An equivalent cut formulation.\ } Given an instance of \steinerc{p}{q}, we define $\mathcal{S} = \{S \subset V \mid \exists i, |S \cap \{s_i, t_i\}| = 1, |\delta(S)| \leq p+q-1\}$ as the set of \emph{critical} (terminal-separating) cuts. 
Our problem can be equivalently formulated as the following cut-based integer program, in which we have to cover all critical cuts:
\begin{alignat*}{3} 
  \text{min} \quad &\sum_{e \in E} c_e x_e  &   \notag \\
    \text{s.t.} \quad& \sum_{e \in \delta(S)} x_e \geq p  & \quad \forall S \in \mathcal{S} \tag{CutIP}\\  
    & x_e \in \{0,1\}  & \forall e \in E. \notag    
\end{alignat*}

\begin{restatable}{proposition}{cutFormulation}
    \label{prop:cut formulation}
    \textup{(CutIP)} characterizes the feasible solutions of \steinercshort{p}{q}.
    Moreover, a solution is feasible if and only if any critical cut contains at least~$p$ protected edges.
\end{restatable}

\begin{proof}
    We show that an edge set $X$ is a feasible solution if and only if for any vertex set $S \in \mathcal{S}$, $|X \cap \delta(S)| \geq p$.
   Consider a feasible solution $X$ and suppose that there is some $S \in \mathcal{S}$ with $|\delta(S)| \leq p+q-1$ and $|\delta(S) \cap X| < p$. 
   Then, after removing no more than $q$ edges from $\delta(S) \setminus X$, the remaining graph has a cut with less than $p$ edges. 
   Further, this cut separates some terminal pair. 
   Thus, $X$ is not a feasible solution.

    Suppose for all $S \in \mathcal{S}$ we have $|\delta(S) \cap X| \geq p$. 
    We show that after removing at most $q$ unprotected edges, the remaining graph is still $p$-edge-connected between any terminal pair. 
    For any cut $\delta(S)$ with $|\delta(S)| \geq p+q$, there are at least $p$ remaining edges since we remove at most $q$ edges. Fix any terminal pair $s,t$ and any edge set $D \subseteq (E \setminus X)$ with $|D| \leq q$. We show that $|\delta(S) \setminus D| \geq p$ for any $s$-$t$-cut $S$, which implies $p$-edge-connectivity between $s,t$. If $|\delta(S)| \geq p+q$, then $|\delta(S) \setminus D| \geq p$. If $|\delta(S)| \leq p+q-1$, then $|\delta(S) \cap X| \geq p$ by the constraint of (CutIP). Thus it also holds that $|\delta(S) \setminus D| \geq p$.
\end{proof}
 
Given a partial solution $X \subseteq E$, we call a cut \emph{safe} (w.r.t.\ $X$) if it is not critical or it contains at least~$p$ edges in~$X$. 
Otherwise, we call it \emph{unsafe}.

\textbf{NP-hardness.\ }
In addition to the aforementioned complexity observations, we now show that \globalcshort{1}{q} is \NP-hard, even in the unweighted setting where we protect a minimum number of edges.
Observe that any spanning tree of $G$ is a feasible solution, which implies $\opt \leq |V|-1$. 
However, we show that it is \NP-complete to distinguish whether $\opt = |V|-1$ or $\opt < |V|-1$, by a reduction from the largest bond problem \cite{duarte2021computing}. 
Therein, we are given an undirected graph $G=(V, E)$ and an integer $k \geq 1$. 
A bond is an edge set represented by $\delta(S)$ for some $S \subset V$ with both $G[S]$ and $G[V \setminus S]$ being connected.
The task is to decide whether there is a bond of size at least $k$.

We outline the idea for the hardness proof as follows.
Given an instance of the largest bond problem, we reduce it to an instance of \globalcshort{1}{q} using the same graph with $q:=k-1$.
If there is a bond $\delta(S)$ of size at least $k=q+1$, then protecting a spanning tree of $G[S]$ and $G[V \setminus S]$ is feasible, as the cut $\delta(S)$ is not critical, which implies $\opt < |V|-1$.
If $\opt < |V|-1$, then the protected edges in the optimal solution induce multiple connected components. 
We can find a bond of size at least $q+1$ by contracting the connected components induced by the protected edges and computing the minimum cut of the resulting graph.

\begin{restatable}{theorem}{hardnessglobaloneq}
\label{thm:hardness:1q}
    Unweighted \globalc{1}{q} is \NP-hard. %even when all edges have unit cost.
\end{restatable}
\begin{proof}
    Given an instance of the largest bond problem, we construct an instance of \globalc{1}{q} using the same graph with $q=k-1$. 
    
    If there is a bond $\delta(S)$ of size at least $k$, then the optimal solution value of the \globalc{1}{q} instance is no more than $|V|-2$, as we can simply protect a spanning tree of $G[S]$ and a spanning tree of $G[V\setminus S]$), which exist since $\delta(S)$ is a bond.

    If there is no bond of size at least $k$, we claim that the optimal solution of the instance of \globalc{1}{q} must be a spanning tree using $|V|-1$ edges. 
    Suppose the optimal solution is not a spanning tree, and it consists of connected components $S_1, S_2, \dots, S_t$, $t \geq 2$. 
    After contracting $S_1, S_2, \dots, S_t$, the graph has to be $k$-edge-connected, by feasibility.
    Let $G'$ be this graph.
    Note that in any graph there must be a minimum cut $Y \subseteq E$ such that $G \setminus Y$ consists of exactly $2$ connected component. 
    Hence, there is also such a minimum cut $Y$ in $G'$ and this cut has size at least $k$, as $G'$ is $k$-edge-connected. 
    But then $Y$ corresponds to a bond of size at least $k$ in the original graph, a contradiction.
\end{proof}

\section{Exact Algorithms for small $q$} 
\label{sec:exact-algos}
In this section we design three polynomial-time exact algorithms for different cases depending on $p$ and $q$, i.e., we prove Theorems~\ref{thm:poly-exact-p-1}, \ref{thm:poly-exact-1-2}, and~\ref{thm:poly-exact-2-2}, which together imply~\Cref{thm:main:polytime}.
 
\subsection{\steinerc{p}{1}} \label{sec:p_1}
    To give some intuition, we first show a simple algorithm for \steinerc{p}{1}. By \Cref{prop:cut formulation}, an instance is feasible if and only if there is no terminal-separating cut of size less than $p$. 
    Hence, from now on we assume the instance is feasible.
    
    The set of critical cuts is given by $\mathcal{S} = \{S \subset V \mid \exists i, |S \cap \{s_i, t_i\}| = 1, |\delta(S)| \leq p\}$.
    Hence, any feasible solution must contain \emph{all} edges of any critical cut.
    Therefore, the only inclusion-wise minimal solution consists of all edges in any terminal-separating cut of size $p$ and it remains to find all such edges. 
    To this end, we assign different \emph{capacities} to protected edges and unprotected edges such that any safe cut has a strictly larger capacity than that of any unsafe cut. The algorithm works as follows.
    
    {\bf Algorithm 1.}\ 
    Let $X$ be the current partial solution; initially $X = \emptyset$.
    %Starting from $X= \emptyset$, we iteratively add edges to $X$.}
    In each iteration, we set the capacity of the edges to  %The capacities are 
    $\frac{p+1}{p}$ for all $e \in X$ and $1$ otherwise.
    For every terminal pair $s,t$, we solve the Minimum $s$-$t$-Cut Problem using standard polynomial-time algorithms. 
    If we find a terminal-separating cut of capacity less than $p+1$, then this defines an unsafe critical cut $\delta(S)$ and we protect all edges in it, i.e., we add $\delta(S)$ to~$X$ and repeat.
    %Otherwise, 
    If each terminal-separating cut has capacity at least $p+1$, output $X$.
    %We show that this algorithm outputs an optimum solution.
    
\begin{restatable}{theorem}{polytimeexactpone}
\label{thm:poly-exact-p-1}
Algorithm 1 is a polynomial-time exact algorithm for \steinerc{p}{1}.
\end{restatable}
\begin{proof}
    Algorithm~1 runs obviously in polynomial time. 
    Note that we can decide the feasibility of the given instance by enumerating terminal pairs and checking whether there is a terminal-separating cut of size less than $p$. %In the following 
    We now assume there is none and the instance is feasible. 
    
    Let $X \subseteq E$ be a partial solution.  
    We claim that the capacity function in Algorithm~1 distinguishes safe and unsafe cuts with respect to $X$.
    Specifically, a cut is unsafe if and only if its capacity is strictly less than $p+1$. 
    By the feasibility of the instance, any terminal-separating cut has at least $p$ edges. 
    Let $C$ be any terminal-separating cut. 
    If $|C| > p$ or $C \subseteq X$, then the capacity of $C$ is at least $p+1$. 
    If $|C| = p$ and $|C \cap X| < p$, then its capacity is smaller than $p+1$. 
    Thus we can find an unsafe terminal-separating cut by enumerating terminal pairs $s, t$ and computing a minimum $s$-$t$ cut with respect to the capacity function. 
    By the preceding discussion, Algorithm 1 finds an optimum solution in polynomial time.
\end{proof}

\subsection{\steinerc{1}{2}}
In this subsection we present a polynomial-time algorithm for \steinerc{1}{2}.  
The set of critical cuts is $\mathcal{S} = \{S \subset V \mid \exists i, |S \cap \{s_i, t_i\}| = 1, |\delta(S)| \leq 2\}$.
Hence, we distinguish between bridges and $2$-edge-cuts. 
We first show that we can reduce to the case that the input graph $G$ is $2$-edge-connected.

Given any bridge $e$ of $G$, if $e$ separates some terminal pair, any feasible solution has to include $e$. 
In this case, we pay $c(e)$ and consider the new instance defined by $G/e$. 
If there is no such terminal pair, then any inclusion-wise minimal feasible solution should not include~$e$, which implies that we can delete $e$ and consider the two connected components of $G - e$ individually. 
As a result, we can assume that the input graph $G$ is $2$-edge-connected.
Note that if the graph is $3$-edge-connected, then there is no critical cut and we are done.

Given a terminal-separating $2$-edge-cut $\{e_1, e_2\}$ of $G$, at least one of $e_1$ and $e_2$ has to be contained in any feasible solution. However, deciding which edge to protect is non-trivial. 
%Let $\{e_1, e_2\}$ be a terminal-separating cut. 
We show how to further decompose our instance into smaller and independent instances according to the following structural lemma. See \Cref{fig:decompositiona} for an illustration.

\begin{lemma}\label{decomposition} 
    Given an undirected graph $G$ which is $2$-edge-connected but not $3$-edge-connected, and a $2$-edge-cut $\{e_1, e_2\}$ of $G$, there is a polynomial-time algorithm to decompose $G$ into disjoint $2$-edge-connected subgraphs $G_1, \dots, G_k$ such that after contracting $G_1, \dots, G_k$ the resulting graph $G/\bigcup_{i=1}^kG_i$ forms a cycle and $e_1, e_2$ belong to this cycle.
\end{lemma}

\begin{proof}
    Consider the graph $G' := G \setminus \{e_2\}$, which is connected but not $2$-edge-connected.
    Let $G''$ arise from $G'$ by contracting each $2$-edge-connected component.
    Note that $G''$ is isomorphic to a tree. 
    Since $G=G'\cup \{e_2\}$ is $2$-edge-connected, $G''$ must be a path. 
    Further, $e_2$ connects the two end-vertices of the path and $e_1$ is a path edge ($e_1$ is a bridge of $G''$). Let the nodes of the path $G''$ be $v_1, \dots, v_k$ and for $1 \leq i \leq k$ let $G_i$ be the $2$-edge-connected component represented by $v_i$, respectively. 
    We conclude that $G/\bigcup_{i=1}^kG_i$ forms a cycle and $e_1, e_2$ belong to this cycle.
\end{proof}

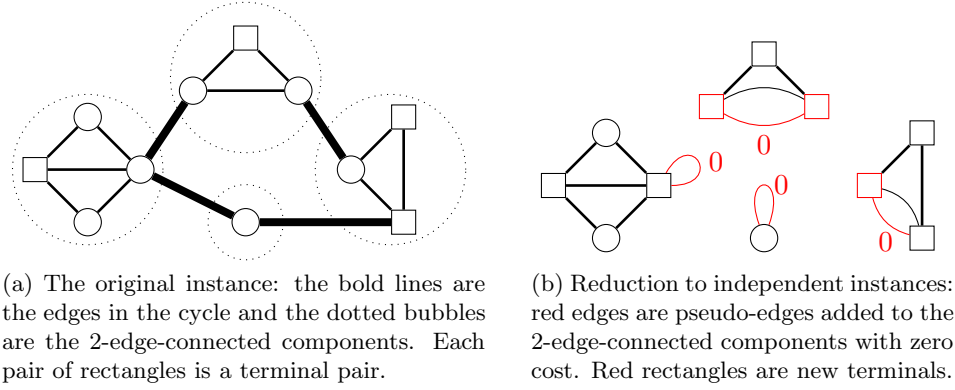
\begin{figure}[t]
\centering
\begin{subfloat}[The original instance: the bold lines are the edges in the cycle and the dotted bubbles are the $2$-edge-connected components. Each pair of rectangles is a terminal pair.] {
\centering
    \begin{tikzpicture}[scale=0.35]
    \node[draw=black,rectangle,minimum size =0.32cm] (a) at (0,0) {};
    \node[draw=black,circle] (b) at (2,2) {};
    \node[draw=black,circle] (c) at (2,-2) {};
    \node[draw=black,circle] (d) at (4,0) {};

    \draw[draw=black,line width=1pt] (a) -- (b);
    \draw[draw=black,line width=1pt] (a) -- (d);
    \draw[draw=black,line width=1pt] (a) -- (c);
    \draw[draw=black,line width=1pt] (b) -- (d);
    \draw[draw=black,line width=1pt] (c) -- (d);

    \node[draw=black,circle,minimum size = 2cm,dotted] (G1) at (2,0) {};

    \node[draw=black,circle] (e) at (6,3) {};
    \node[draw=black,circle] (f) at (10,3) {};
    \node[draw=black,rectangle,minimum size =0.32cm] (g) at (8,5) {};

    \draw[draw=black,line width=3pt] (d) -- (e);
    \draw[draw=black,line width=1pt] (e) -- (f);
    \draw[draw=black,line width=1pt] (f) -- (g);
    \draw[draw=black,line width=1pt] (e) -- (g);

    \node[draw=black,circle,minimum size = 2.0cm,dotted] (G1) at (8,3.5) {};

    \node[draw=black,circle] (h) at (12,0) {};
    \node[draw=black,rectangle,minimum size =0.32cm] (i) at (14,2) {};
    \node[draw=black,rectangle,minimum size =0.32cm] (j) at (14,-2) {};

    \draw[draw=black,line width=3pt] (f) -- (h);
    \draw[draw=black,line width=1pt] (h) -- (i);
    \draw[draw=black,line width=1pt] (i) -- (j);
    \draw[draw=black,line width=1pt] (h) -- (j);

    \node[draw=black,circle,minimum size = 2cm,dotted] (G1) at (13.5,0) {};

     \node[draw=black,circle] (k) at (8,-2) {};
     \node[draw=black,circle,minimum size = 1cm,dotted] (G1) at (8,-2) {};
    
     \draw[draw=black,line width=3pt] (d) -- (k);
     \draw[draw=black,line width=3pt] (k) -- (j);
\end{tikzpicture}
\label{fig:decompositiona} 
}
\end{subfloat}
\hspace{10pt}       
\begin{subfloat}[%Breaking into 
Reduction to independent instances: red edges are pseudo-edges added to the $2$-edge-connected components with zero cost. Red rectangles are new terminals.] {
\centering
    \begin{tikzpicture}[scale=0.35]
    \node[draw=black,rectangle,minimum size =0.32cm] (a) at (0,0) {};
    \node[draw=black,circle] (b) at (2,2) {};
    \node[draw=black,circle] (c) at (2,-2) {};
    \node[draw=black,rectangle,minimum size =0.32cm] (d) at (4,0) {};
    \path[every loop/.style={min distance=10mm,in=0,out=60,looseness=10}] (d) edge[loop right,red] node[red]{$0$} (d);

    \draw[draw=black,line width=1pt] (a) -- (b);
    \draw[draw=black,line width=1pt] (a) -- (d);
    \draw[draw=black,line width=1pt] (a) -- (c);
    \draw[draw=black,line width=1pt] (b) -- (d);
    \draw[draw=black,line width=1pt] (c) -- (d);

    %\node[draw=black,circle,minimum size = 3cm,dotted] (G1) at (2,0) {};

    \node[draw=red,rectangle,minimum size =0.32cm] (e) at (6,3) {};
    \node[draw=red,rectangle,minimum size =0.32cm] (f) at (10,3) {};
    \node[draw=black,rectangle,minimum size =0.32cm] (g) at (8,5) {};

    \path[]
        (e) edge [bend left] node {} (f)
            edge [bend right,red] node[red,below] {$0$} (f);
    \draw[draw=black,line width=1pt] (f) -- (g);
    \draw[draw=black,line width=1pt] (e) -- (g);

    %\node[draw=black,circle,minimum size = 2.8cm,dotted] (G1) at (8,4.5) {};

    \node[draw=red,rectangle,minimum size =0.32cm] (h) at (12,0) {};
    \node[draw=black,rectangle,minimum size =0.32cm] (i) at (14,2) {};
    \node[draw=black,rectangle,minimum size =0.32cm] (j) at (14,-2) {};

    \draw[draw=black,line width=1pt] (h) -- (i);
    \draw[draw=black,line width=1pt] (i) -- (j);

    \path[]
        (h) edge [bend left] node {} (j)
            edge [bend right,red] node[red,below] {$0$} (j);

    %\node[draw=black,circle,minimum size = 3cm,dotted] (G1) at (13.5,0) {};

     \node[draw=black,circle] (k) at (8,-2) {};
     \path[every loop/.style={min distance=10mm,in=70,out=110,looseness=15}] (k) edge[loop right,red] node[red,right]{$0$} (k);
     
\end{tikzpicture}
\label{fig:decompositionb} 
}       
\end{subfloat}
\caption{Illustration of the decomposition (\Cref{decomposition}, \Cref{lemma: decomposition}).}
\label{fig:decomposition}

\end{figure}
Given a decomposition as in \Cref{decomposition}, we claim that the problem reduces to solving certain subproblems defined by $G_1, \dots, G_k$ (plus some additional pseudo-edge for each component) and the subproblem restricted to the cycle $C$. 
To do so, we view our problem as finding a minimum-cost edge set that intersects all $2$-edge-cuts. 
Observe that any inclusion-wise minimal $2$-edge-cut is either %(i) 
two edges on the cycle $C$, or %(ii) 
two edges in $G_i$ for some~$i$.
%This implies that 
Hence, we can solve our problem by solving $(i)$ the subproblem defined by $2$-edge-cuts on the cycle~$C$ and $(ii)$ the subproblems defined by $2$-edge-cuts in each component $G_i$ separately. 

The subproblem $(i)$ %on the cycle 
is a %minimum-cost 
Steiner Forest problem on a cycle. 
This follows from the observation that any feasible solution must contain a path consisting of only protected edges between each terminal pair. 
We can solve the min-cost Steiner Forest problem on a cycle by enumerating which cycle-edge is not in the optimum solution and breaking the cycle into a path.
On a path, the solution is the union of the unique paths between the terminal pairs. 
Then, we recursively solve the subproblems $(ii)$ in each $G_i$. 
However, we cannot simply recurse on $G_i$ since a $2$-edge-cut of $G_i$ may not be a $2$-edge-cut of $G$. Instead, we recourse on a new graph obtained by adding a zero-cost edge $e_i$ to $G_i$. 
This edge $e_i$ connects the two vertices that are incident to the edges of $C$, which represents the connection in $G_i$ between these vertices via the cycle $C$.
We formalize this idea in the following lemma (See \Cref{fig:decompositionb}).

\begin{restatable}{lemma}{cycleDecomposition}\label{lemma: decomposition}
    Given a decomposition as in \Cref{decomposition}, an optimum solution to \steinercshort{1}{2} can be obtained by combining optimum solutions of the following subproblems:
    \begin{itemize}
        \item[(i)] protect a minimum-cost edge set that intersects with any terminal-separating $2$-edge-cut on the cycle, and
        \item[(ii)] for each $G_i$, let $u_i,v_i \in V(G_i)$ be the two vertices incident to the two edges in the cycle. Solve the problem on $G_i' = G_i \cup \{(u_i,v_i)\}$ with $c_{(u_i,v_i)} = 0$. 
        Keep the terminal pairs with both terminals in $G_i$. For terminal pairs $(s_i,t_i)$ with $s_i \in G_i, t_i \notin G_i$, replace it with $(s_i, u_i), (s_i,v_i)$.
    \end{itemize}
\end{restatable}

\begin{proof}
    We first show that given any feasible solution $X$ of $G$, the corresponding edges of $X$ on each subproblem is a feasible solution for the subproblem. 
    For the cycle subproblem (ii) this is trivial. 
    For any subproblem $G_i' = G_i \cup \{(u_i, v_i)\}$, we show $X \cap G_i \cup \{(u_i, v_i)\}$ is a feasible solution. 
    Observe that any $2$-edge-cut $C$ in $G'$ cannot contain the edge $\{(u_i, v_i)\}$ since $G_i$ is $2$-edge-connected. 
    Thus, $C$ must also be a $2$-edge-cut in $G$. 
    If $C$ separates some terminal pair in $G_i'$, so does it in $G$, which implies $C \cap X \neq \emptyset$. Therefore, each terminal-separating $2$-edge-cut of $G_i'$ is safe.

    Given feasible solutions of the subproblems, we show how to obtain a feasible solution of $G$ without increasing the cost. 
    Let $X$ be the edges protected in the subproblems except the new edges $(u_i,v_i)$. 
    Thus, the cost of $X$ is at most the sum of the cost of the solutions to the individual subproblems. 
    It remains to argue that $X$ is feasible for $G$. 
    Let $C$ be any terminal-separating $2$-edge-cut of $G$. 
    If $C$ is on the cycle, by the feasibility of the subproblem on the cycle, $C \cap X \neq \emptyset$. 
    If $C$ is in $G_i$ for some $i$, $C$ must also be a terminal-separating $2$-edge-cut of $G_i'$. 
    Thus $C \cap X \neq \emptyset$. 
    We conclude that $X$ is feasible for~$G$. 
\end{proof}

%
%In summary, the algorithm works as follows.

\medskip 
{\bf Algorithm 2.\ } We first protect terminal-separating bridges, contract them, and consider the $2$-edge-connected components separated by non-terminal-separating bridges individually. 
This reduces to the case that $G$ is $2$-edge-connected.
Then, as long as we find a terminal-separating $2$-edge-cut (which is the only type of critical cut), we decompose the problem into a subproblem on a cycle and a collection of subproblems in smaller $2$-edge-connected components. 
Then we recursively solve the individual subproblems.
The decomposition stops either if $G$ is $3$-edge-connected (and hence we are done as there is no critical cut) or each component on the cycle consists of a single vertex, i.e., if $G$ is a cycle.
The cycle case is solved by enumerating which edge of the cycle is \emph{not} contained in an optimum solution and then solving a Steiner Forest problem on a path where the optimal solution is trivial. Among all such solutions, we output the one with minimum cost.
%We obtain the following result.

\begin{theorem}
\label{thm:poly-exact-1-2}
    Algorithm 2 is a polynomial-time exact algorithm for \steinerc{1}{2}.  
\end{theorem}
 
We remark that Bienstock and Diaz~\cite{DBLP:journals/siamcomp/BienstockD93} studied a special case of \steinercshort{1}{q}.
They showed that it is \NP-hard when $q=8$ and they conjectured the \NP-hardness for $q= 3$.

Interestingly, \globalc{1}{2} admits an easier algorithm. 
We view the problem as finding a minimum-cost edge set hitting all $2$-edge-cuts, which reduces to a special case of Minimum Weighted Vertex Cover. 
Therein, each edge $e$ of $G$ corresponds to a vertex $v_e$ in the Vertex Cover instance $G'$ and there is an edge between two vertices $v_e$ and $v_{e'}$ if and only if $\{e, e'\}$ forms a $2$-edge-cut in $G$.
%2$-edge-cut corresponds to a 
The Vertex Cover instance $G'$ has a special structure with each connected component being a complete graph. 
To see this observe that, if both $\{e_1, e_2\}$ and $\{e_1, e_3\}$ are $2$-edge-cuts, then $\{e_2, e_3\}$ is also a $2$-edge-cut~(\cite[Lemma~2.37]{DBLP:books/cu/NI2008}). 
The optimal vertex cover solution in a complete graph is trivial: select all vertices except the largest-weighted one. 
We conclude with the following result.
\begin{lemma} 
The greedy algorithm that selects from each $2$-edge-cut the cheaper edge solves \globalc{1}{2}.
\end{lemma}

\subsection{\globalc{p}{q} when $p,q \leq 2$}
We now present a polynomial-time algorithm for \globalcshort{2}{2}.
Note that for all other $p, q \leq 2$, our previous results imply a polynomial-time algorithm for \globalcshort{p}{q}.
We outline our algorithm as follows.
By \Cref{prop:cut formulation}, we can assume the input graph $G$ to be $2$-edge-connected, as otherwise, the instance is infeasible. 
Further, a feasible solution contains at least two edges in each $2$-edge-cut and in each $3$-edge-cut.
We first show that if there is some $2$-edge-cut, it is equivalent to solving two smaller independent instances (\Cref{lem:global22:3-connected}). 
Hence, we can assume that the input graph $G$ is $3$-edge-connected.
Then we represent all the $3$-edge-cuts using a standard tree representation~\cite{dinitz1973structure,DBLP:conf/soda/HeHS24a}
and it remains to solve a weighted multi-commodity flow problem on the tree (introduced formally later, \Cref{lem:global22:reduction-multi-commodity-flow}).
Finally, we solve the weighted multi-commodity flow problem via dynamic programming (\Cref{lem:global22:dp}).

Suppose $G$ is not $3$-edge-connected and there is some $2$-edge-cut $\{e_1, e_2\}$, i.e., $G\setminus \{e_1, e_2\}$ consists of $2$ connected components $G_1, G_2$. 
Let $e_1 = (u_1, u_2)$ with $u_1 \in G_1$ and $u_2 \in G_2$, $e_2 = (v_1, v_2)$ with $v_1 \in G_1$ and $v_2 \in G_2$ (see \Cref{fig:lem:global22:3-connected}). 
We create two new instances: $I_1$ on $G_1 \cup \{(u_1,v_1)\}$ where $(u_1,v_1)$ is an edge with zero cost and $I_2$ on $G_2 \cup \{(u_2,v_2)\}$, where $\{(u_2,v_2)\}$ has zero cost.
We show that it suffices to solve $I_1, I_2$ independently and combine their solutions to get a solution to the original instance.

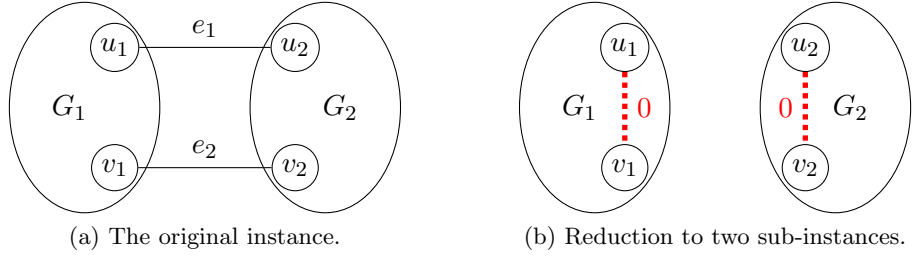
\begin{figure}[t]\label{fig: split}
\centering
\begin{subfloat}[The original instance.] {
\centering
    \begin{tikzpicture}[scale=0.4]
        \draw (-4,0) ellipse (2.5cm and 3.5cm);
        \node at (-4.5,0) {$G_1$};
        \draw (4,0) ellipse (2.5cm and 3.5cm);
        \node at (4.5,0) {$G_2$};
        
        \node[draw=black,circle,minimum size = 0.5cm,inner sep=2pt] (u1) at (-3,2) {$u_1$};
        \node[draw=black,circle,minimum size = 0.5cm,inner sep=2pt] (v1) at (-3,-2) {$v_1$};
        
        \node[draw=black,circle,minimum size = 0.5cm,inner sep=2pt] (u2) at (3,2) {$u_2$};
        \node[draw=black,circle,minimum size = 0.5cm,inner sep=2pt] (v2) at (3,-2) {$v_2$};
        
        \draw[black]
        (u1) -- node[above] { $e_1$} (u2)
        (v1) -- node[above] { $e_2$} (v2);
    \end{tikzpicture}
}
\end{subfloat}
\hspace{30pt}       
\begin{subfloat}[Reduction to two sub-instances.] {
\centering
    \begin{tikzpicture}[scale=0.4]
        \draw (-4,0) ellipse (2.5cm and 3.5cm);
        \node at (-4.5,0) {$G_1$};
        \draw (4,0) ellipse (2.5cm and 3.5cm);
        \node at (4.5,0) {$G_2$};
        
        \node[draw=black,circle,minimum size = 0.5cm,inner sep=2pt] (u1) at (-3,2) {$u_1$};
        \node[draw=black,circle,minimum size = 0.5cm,inner sep=2pt] (v1) at (-3,-2) {$v_1$};
        
        \node[draw=black,circle,minimum size = 0.5cm,inner sep=2pt] (u2) at (3,2) {$u_2$};
        \node[draw=black,circle,minimum size = 0.5cm,inner sep=2pt] (v2) at (3,-2) {$v_2$};
        
        \draw[line width=2pt,red,dotted]
        (u1) -- node[right] {$0$ } (v1)
        (u2) -- node[left] { $0$} (v2);
    \end{tikzpicture}
}       
\end{subfloat}
\caption{Illustration of \Cref{lem:global22:3-connected}: %it's equivalent to
Equivalence to solving two new instances on $G_1 \cup \{(u_1,v_1)\}$ where $(u_1,v_1)$ is an edge with zero cost and $G_2 \cup \{(u_2,v_2)\}$ where $\{(u_2,v_2)\}$ has zero cost.}
\label{fig:lem:global22:3-connected}

\end{figure}
\begin{restatable}{lemma}{lemglobaltwotwothreeconnected}
\label{lem:global22:3-connected}
    $\opt(I) = c(e_1) + c(e_2)+\opt(I_1)+\opt(I_2)$.
\end{restatable}

\begin{proof}
    Given a feasible solution $X$ of $I$, we show that $X_1 = X \cap E(G_1) + (u_1,v_1)$ and $X_2 = X \cap E(G_2) + (u_2,v_2)$ are feasible solutions for $I_1$ and $I_2$, respectively, implying $c(X) = c(e_1) + c(e_2)+c(X_1)+c(X_2)$ and $\opt(I) \geq c(e_1)+c(e_2)+\opt(I_1)+\opt(I_2)$. 
    We show feasibility for $X_1$; the feasibility for $X_2$ is analogous.
    It suffices to show that for any critical cut in $G_1 + (u_1,v_1)$, at least $2$ edges are protected.
    %and the proof for $I_2$ is symmetric. 
    Consider any critical cut $C$ of $G_1 + (u_1,v_1)$. 
    If $C$ does not contain the new edge $(u_1,v_1)$, then it is also a critical cut of $G$ and therefore this cut is safe. 
    Hence, we assume that $C$ contains the new edge $(u_1,v_1)$ and let $C= \{f_1, f_2, (u_1,v_1) \}$.
    Note that $f_2$ might not exist if $|C|=2$.
    We show that $\{ f_1, f_2, e_1\}$ is a critical cut in $G$. 
    Consider $G'_1 = (G_1 \setminus \{ f_1, f_2 \}) + (u_1,v_1)$ and observe that $(u_1,v_1)$ is a bridge in $G'_1$, as $C$ is a cut in $G_1 + (u_1,v_1)$.
    Hence, $u_1$ and $v_1$ are only connected in $G_1'$ via the edge $(u_1,v_1)$.
    This means that also in $G \setminus \{f_1, f_2 \}$, any path connecting $u_1$ and $v_1$ must use $e_1$.
    Hence, $e_1$ is a bridge in $G \setminus \{f_1, f_2 \}$ and therefore $\{ f_1, f_2, e_1\}$ is a critical cut in $G$.
    Thus $C \setminus (u_1,v_1) + e_1$ contains at least $2$ protected edges, which implies $C \setminus (u_1,v_1)$ contains at least $1$ protected edge. 
    Since $(u_1,v_1)$ is protected, $C$ has at least $2$ protected edges and is safe.

    On the other hand, given solutions $X_1$ and $X_2$ of $I_1$ and $I_2$, respectively, we show that $X = X_1 \cup X_2 + e_1 + e_2$ is feasible for $I$, implying $c(X) = c(e_1) + c(e_2)+c(X_1)+c(X_2)$ and $\opt(I) \leq c(e_1)+c(e_2)+\opt(I_1)+\opt(I_2)$.
    Let $C$ be any critical cut of $G$. 
    Without loss of generality, we assume that $C$ is inclusion-wise minimal, i.e., it does not contain any smaller critical cut. 
    We distinguish the following three cases. 
    First, assume that $C \subseteq G_1$. 
    Then, $C$ must also be an edge cut of $G_1 \cup (u_1,u_2)$ and therefore $C$ is safe. 
    The case $C \subseteq G_2$ is analogous.
    For the second case, we assume that the first case does not apply and further assume that $C \cap \{e_1, e_2\} = \emptyset$. 
    Hence, either $|C \cap G_1| =1 $ or $|C \cap G_2|=1$. 
    Without loss of generality assume $|C \cap G_1|=1 $. 
    Observe that $(C \cap G_2) +(u_2, v_2)$ is a critical edge cut in $G_2 + (u_2, v_2)$ and $(C \cap G_1) + (u_1, v_1)$ is a critical edge cut in $G_1 + (u_1, v_1)$. 
    Further, by feasibility of $X_1$ and $X_2$, the only edge in $C \cap G_1$ and at least one of the two edges in $C \cap G_2$ must be protected, which implies $C$ contains at least $2$ protected edges and is safe.
    In the third and final case, we assume that none of the previous cases apply and further assume that $C$ contains either $e_1$ or $e_2$. 
    Any cut containing both $e_1$ and $e_2$ is safe, as both are protected in $X$. 
    Without loss of generality assume $e_1 \in C$. 
    We claim that either $C - e_1 \subseteq E(G_1)$ or $C -e_1 \subseteq E(G_2)$. 
    Otherwise, $C$ contains one edge $e_3$ in $G_1$ and one edge $e_4$ in $G_2$. 
    Observe that $\{e_1, e_3\}$ is a $2$-edge-cut of $G$, which contradicts the fact that $C$ is inclusion-wise minimal.
    If $C - e_1 \subseteq E(G_1)$, then similar to the first part of this proof one can show that $C - e_1 + (u_1, v_1)$ is a cut in $G_1 + (u_1, v_1)$. 
    Hence, $|X_1 \cap C| \geq 1$ and therefore, $|X \cap C| \geq 2$, as $e_1 \in X$.
    The case $C -e_1 \subseteq E(G_2)$ is analogous and hence this concludes the proof.
\end{proof}

By repeating the above process, we end up with a $3$-edge-connected graph.
The following tree representation gives us a clear structure about all the $3$-edge-cuts and it can be computed in near-linear time~\cite{DBLP:conf/soda/HeHS24a}.

\begin{definition} (Tree representation of min cuts \cite{dinitz1973structure})
Let $G=(V,E)$ be an undirected graph and suppose the capacity of its minimum cut is an odd number $k$. There is a polynomial-time algorithm that constructs a rooted tree $T=(U,F)$ together with a (not necessarily surjective) mapping $\phi: V \rightarrow U$. Further, there is a one-to-one correspondence between any $k$-edge-cut of $G$ and $f \in F$ as follows. For any $f \in F$, let $T_f$ be the subtree of $T$ beneath~$f$ and let $V(T_f) = \{v \in V \mid \phi(v) \in T_f\}$. Then for any tree edge $f\in F$, $\delta_G(V(T_f))$ defines a $k$-edge-cut of $G$. %On the other hand, 
For any $k$-edge-cut $C$ of $G$, there is some tree edge $f \in F$ such that $C = \delta_G(V(T_f))$. 
\end{definition}

Given this tree representation, we show now that our problem \globalc{2}{2} reduces to the weighted multi-commodity flow problem on a tree. This problem is defined 
as follows: given a tree $T$, a set of paths $\mathcal{P}$ on the tree and a weight function $w: \mathcal{P} \rightarrow \mathbb{R}$, find a subset of pairwise edge-disjoint paths with maximum total weight.

\begin{lemma} 
\label{lem:global22:reduction-multi-commodity-flow}
    When the input graph $G$ is $3$-edge-connected, \globalc{2}{2} reduces to the weighted multi-commodity flow problem on a tree.
\end{lemma}

\begin{proof}
    
A solution $X \subseteq E$ is feasible if and only if it contains at least $2$ edges in each $3$-edge-cut. 
Equivalently, for each $3$-edge-cut at most one edge is unprotected. 
We consider this complement problem in which we want to find a set of maximum-weight edges $\Bar{X}$ such that each $3$-edge-cut contains at most one edge of $\Bar{X}$.
We use the standard tree representation~\cite{dinitz1973structure} of all the $3$-edge-cuts of~$G$, in which each $3$-edge-cut of the original graph is represented by an edge in the tree.
Given a tree representation $T=(U,F)$ and $e=(u,v) \in E$, define $P_e$ as the path on $T$ between $\phi(u)$ and $\phi(v)$ and let the weight of $P_e$ be the cost of $e$.
Observe that every $3$-edge-cut containing $e$ corresponds to a tree edge on $P_e$ and vice versa.
Therefore, a solution $X \subseteq E$ is feasible if and only if the set of paths $\{P_e \mid e \in \Bar{X} = E\setminus X \}$ are pairwise edge-disjoint.
Hence finding the optimal $X$ reduces to finding a set of edge-disjoint paths on $T$, maximizing the total weight, which is the weighted multi-commodity flow problem.
\end{proof}

\begin{remark}
    We can prove that \Cref{lem:global22:reduction-multi-commodity-flow} holds more generally for any even $p$: If $G$ is $(p+1)$-edge-connected, \globalcshort{p}{2} reduces to the weighted multi-commodity flow problem on a tree for any even $p$. However, to solve \globalcshort{p}{2} using this reduction, we need to reduce the problem to the case where $G$ is $(p+1)$-edge-connected, which remains unclear for $p\geq 4$.
\end{remark}

Garg et al.~\cite{DBLP:journals/algorithmica/GargVY97} considered an unweighted version of multi-commodity flow problem on a tree and obtained an exact polynomial-time greedy algorithm.
However, their arguments do not extend to the weighted case.
We design a dynamic program for the weighted version. 
\begin{lemma}
\label{lem:global22:dp}
    The weighted multi-commodity flow problem on a tree can be solved in polynomial time.
\end{lemma}
\begin{proof}
    
We root the tree at an arbitrary vertex $r$.
%We pick an arbitrary vertex $r$ and root the tree at $r$.
Without loss of generality, we assume there is no path that consists of only one vertex since the selected paths have to be only edge-disjoint.
For any vertex $v$, let $T_v$ be the subtree rooted at vertex~$v$.
For any tree edge $e=(u,v)$ where $u$ is closer to the root, we use $T_e$ to represent the subtree $T_u \setminus T_v$ for short. 
Define the subproblem $\mathcal{I}(T_v)$ in the subtree $T_v$ as follows: From the set of paths completely contained in $T_v$, select a maximum-weight subset of paths that are pairwise edge-disjoint. 
Let $f(T_v)$ be the optimal value of $\mathcal{I}(T_v)$. 
We define $\mathcal{I}(T_e)$ and $f(T_e)$ for each $e \in E(T)$ analogously.
We only show how to compute $f(T_v)$; the computation of $f(T_e)$ is similar.

Fix some vertex $v$ and consider the subproblem $\mathcal{I}(T_v)$.
If $v$ is a leaf, then $f(T_v)=0$. Otherwise, let $z_1, \dots, z_k$ be the children of $v$.
Let $\mathcal{P}_v$ be the set of paths intersecting~$v$.
We say a path $P$ occupies the subtree $T_{z_i}$ if it intersects $T_{z_i}$. 
Our first observation is that each subtree $T_{z_i}$ can be occupied by at most one selected path, as otherwise the edge $(v,z_i)$ is contained in multiple selected paths, which is infeasible.
Since a path in $\mathcal{P}_v$ can occupy either two subtrees $T_{z_i}, T_{z_j}$ for some $1 \leq i < j \leq k$ or only one subtree $T_{z_i}$ for some $i$, we reduce the problem $\mathcal{I}(T_v)$ to an instance $\mathcal{M}(T_v)$ of Maximum Weighted Matching. 
For $\mathcal{M}(T_v)$, we create an auxiliary graph $G(T_v)$ as follows.
For each $i$ with $1 \leq i \leq k$, we create a vertex $u_i$ corresponding to the subtree $T_{z_i}$ and a dummy vertex $u_i'$. 
For each path in $\mathcal{P}_v$, if it occupies $T_{z_i}$ and $T_{z_j}$ for some $i,j$, we create an edge between $u_i$ and $u_j$.
If it only occupies one subtree $T_{z_i}$, we create an edge between $u_i$ and $u_i'$.
We also create an extra edge between $u_i$ and $u_i'$ which represents the case where no selected path occupies $T_{z_i}$.
It is not hard to see that there is a one-to-one correspondence between a feasible choice over $\mathcal{P}_v$ in $\mathcal{I}(T_v)$ and a matching on the auxiliary graph $G(T_v)$. 
It remains to properly set the weights $\omega$ of the edges of $G(T_v)$ such that a maximum-weight matching in $G(T_v)$ corresponds to an optimum solution for $\mathcal{I}(T_v)$.
To do so, we observe that for a given fixed feasible choice over $\mathcal{P}_v$, it remains to solve a collection of subproblems represented by $\mathcal{I}(T_v')$ or $\mathcal{I}(T_e')$ for some $v',e'\in T_v$ and combine their optimal solutions.
Formally, let $P=(v_1, \dots, v_\ell)$ be a path in $\mathcal{P}_v$ where $v = v_m$, $1 \leq m \leq \ell$. 
Assume $P$ occupies two subtrees of $v$, say, $(v_1, \dots, v_{m-1}) \subseteq T_{z_i}$ and $(v_{m+1}, \dots, v_{\ell}) \subseteq T_{z_j}$.
The case $P$ only occupies one subtree of $v$ follows analogously.
Suppose we have selected~$P$. 
Then it is still feasible to select paths completely contained in $T_{z_i}$ and $T_{z_j}$, respectively, as long as they do not intersect $P$.
This implies that the subproblems on $T_{z_i} \setminus E(P)$ and $T_{z_j} \setminus E(P)$ can be decomposed into  $T_{v_1}, T_{(v_1, v_2)}, \dots, T_{(v_{m-2}, v_{m-1})}$ and $T_{v_\ell}, T_{(v_\ell, v_{\ell-1})}, \dots, T_{(v_{m+2}, v_{m+1})}$, respectively.
See Figure~\ref{fig:dp} for an example.
Hence, the gain of selecting $P$ is the sum of the optimal values of these subproblems plus the weight of $P$ itself. 
That is, we set $\omega((u_i, u_j)) \coloneqq w(P) + f(T_{v_1}) + \sum_{k=1}^{m-2} f(T_{(v_k, v_{k+1})}) + f(T_{v_\ell}) + \sum_{k=m+1}^{\ell} f(T_{(v_k, v_{k+1})})$. 
For the extra edge between $u_i$ and $u_i'$, which represents no selected path occupies $T_{z_i}$, we set its weight to $f(T_{z_i})$.
It now easily follows that $f(T_v) = \opt(\mathcal{M}(T_v))$, which can be solved in polynomial time using algorithms for Maximum Weighted Matching~\cite{DBLP:conf/soda/Gabow90}.
\end{proof}
\newcommand{\vertex}[4]{
\node[draw=black,circle,inner sep=1pt,minimum size = 0.6cm] (#3) at (#1,#2) {#4}
}
\begin{figure}[t]
\centering
\begin{tikzpicture}
    \vertex{2}{-0.5}{r}{$r$};
    \vertex{-1}{-1}{v}{$v$};
    \draw[dashed,thick] (r) -- (v);
    
    \vertex{-4.5}{-1.5}{w1}{$z_1$};
    \vertex{-2}{-1.5}{wi}{$z_i$};
    \vertex{0}{-1.5}{wj}{$z_j$};
    \vertex{2.5}{-1.5}{wk}{$z_k$};
    
    \draw[] (v) -- (w1) (v) -- (wk);
    \draw[thick,red] (v) -- (wi) (v) -- (wj);
    \draw[dotted,thick,shorten <= 0.5cm, shorten >= 0.5cm] (w1)--(wi);
    \draw[dotted,thick,shorten <= 0.3cm, shorten >= 0.5cm] (wi)--(wj);
    \draw[dotted,thick,shorten <= 0.3cm, shorten >= 0.5cm] (wj)--(wk);

    \vertex{-2}{-2.5}{u2}{$v_2$};
    \vertex{-2}{-3.5}{u1}{$v_1$};
    \vertex{0}{-2.5}{u6}{$v_6$};
    \vertex{0}{-3.5}{u7}{$v_7$};
    \draw[thick,red] (u1)--(u2)--(wi)
                (u7)--(u6)--(wj);

    \draw[thick] (u1)--(-2.5,-4.8)--(-1.5,-4.8)--(u1);
    \node[] (Tu1) at (-2,-4.5) {$T_{v_1}$ };

    \draw[thick] (u2)--(-5,-4.1)--(-2.6,-4.1)--(u2);
    \node[] (Tu12) at (-3.2,-3.7) {$T_{v_2}\setminus T_{v_1}$ };

    \draw[thick] (wi)--(-6,-2.8)--(-2.8,-2.8)--(wi);
    \node[] () at (-3.4,-2.4) {$T_{z_i}\setminus T_{v_2}$ };

    \draw[thick] (u7)--(-0.5,-4.8)--(0.5,-4.8)--(u7);
    \node[] (Tu7) at (0,-4.5) {$T_{v_7}$ };

    \draw[thick] (u6)--(1,-4.1)--(3.2,-4.1)--(u6);
    \node[] (Tu67) at (1.5,-3.7) {$T_{v_6}\setminus T_{v_7}$ };

    \draw[thick] (wj)--(1,-2.8)--(4,-2.8)--(wj);
    \node[] () at (1.6,-2.4) {$T_{z_j}\setminus T_{v_6}$ };
\end{tikzpicture}
\caption{Illustration of subproblems: consider the red path $(v_1,v_2,v_3=z_i,v_4=v,v_5=z_j,v_6,v_7)$ through $v$. If we select the red path, the subtrees $T_{z_i}$ and $T_{z_j}$ break into $T_{v_1}, T_{(v_1,v_2)} = T_{v_2}\setminus T_{v_1}, T_{(z_i, v_2)} = T_{z_i}\setminus T_{v_2}$ and $T_{(z_j,v_6)}=T_{z_j} \setminus T_{v_6}, T_{(v_6,v_7)}=T_{v_6} \setminus T_{v_7}, T_{v_7}$, respectively. 
They define independent subproblems and their optimal solutions have been computed before we compute $f(T_v)$.}
\label{fig:dp}
\end{figure}
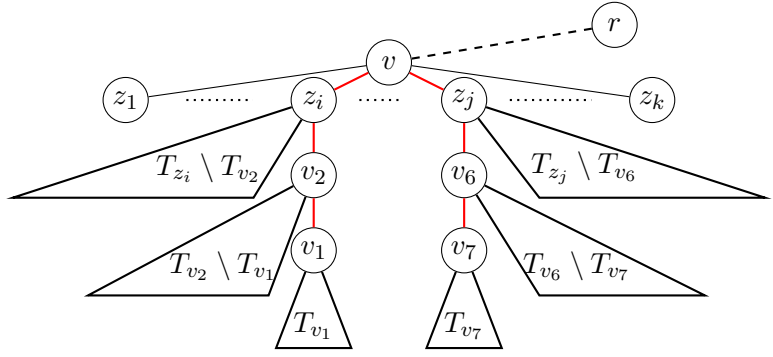

Combining \Cref{lem:global22:3-connected,lem:global22:reduction-multi-commodity-flow,lem:global22:dp}, we conclude with the theorem.
\begin{theorem}
\label{thm:poly-exact-2-2}
There is a polynomial-time exact algorithm for \globalc{2}{2}.
%the $(2,2)$-global-connectivity preservation problem.
%when $p,q \leq 2$.
\end{theorem}

\section{Approximation Algorithms for large $p$ or $q$}
\label{sec:apx-algos}

In this section we provide approximation algorithms and hardness results for large $p$ and $q$. 

\subsection{Approximation for the cases with $p=1$}
We present a primal-dual algorithm for \steinerc{1}{q}. Consider the corresponding linear programming relaxation of (CutIP) and its dual:

\begin{minipage}{0.48\textwidth}
\begin{alignat*}{3}
    \min \quad &\sum_{e \in E} c_e x_e & \\ 
    \text{s.t.} \quad& \sum_{e \in \delta(S)} x_e \geq 1 \quad &\forall S \in \mathcal{S}  \notag \\
    & x_e \geq 0  & \forall e \in E \notag \\ 
\end{alignat*}
\end{minipage}
\hfill
\begin{minipage}{0.48\textwidth}
\begin{alignat*}{3}
    \max \quad &\sum_{S \in \mathcal{S}} y_S & \\ 
    \text{s.t.} \quad& \sum_{S: S \in \mathcal{S}, e \in \delta(S)} y_S \leq c_e \quad &\forall e \in E  \notag \\
    & y_S \geq 0 & \forall S \in \mathcal{S} \notag \\
\end{alignat*}
\end{minipage}

{\bf Algorithm 3.\ }
 We start from a dual solution $\{y_S = 0 \mid  S \in \mathcal{S}\}$ and maintain a partial solution $X \subseteq E$ which is the current protected edge set.
At the beginning, $X:= \emptyset$. We increase the dual variables iteratively and add edges to $X$ whose corresponding dual constraints $\sum_{S: S \in \mathcal{S},e \in \delta(S)} y_S \leq c_e $ become tight.
In each iteration, we pick some $S \in \mathcal{S}$ with $\delta(S) \cap X = \emptyset$ and increase $y_S$. 
%Note that 
Such a vertex set $S$ can be found by enumerating terminal pairs $(s, t)$ and checking whether there is an $s$-$t$-cut of value less than $q+1$ with respect to the following capacity function:
set the capacity of $e$ to $q+1$ if $e \in X$ and to $1$, otherwise.
%and the capacity of other edges as $1$.
We increase $y_S$ until for some edge $e \in \delta(S)$, the dual constraint $\sum_{S: S \in \mathcal{S},e \in \delta(S)} y_S \leq c_e $ is tight.
Then we add $e$ to $X$ and move to the next iteration until $X$ is a feasible solution, which is the case if any terminal-separating cut has a capacity of at least $q+1$ w.r.t.\ the above capacity function.

To bound the cost of $X$, we have
$$
\sum_{e \in X} c_e = \sum_{e \in X} \sum_{S: S \in \mathcal{S},e \in \delta(S)} y_S = \sum_{S \in \mathcal{S}} y_S |\delta(S) \cap X| \leq \sum_{S \in \mathcal{S}} y_S |\delta(S)| \leq q\sum_{S \in \mathcal{S}} y_S \leq q \cdot \opt.
$$

\begin{theorem}
    Algorithm 3 is a polynomial-time $q$-approximation algorithm for \steinerc{1}{q}.
\end{theorem}

%\globalc{1}{q} 
The {\em global} connectivity variant \globalcshort{1}{q} has more symmetry since we do not need to distinguish whether an edge cut is terminal-separating. 
By exploiting the special structure of the family $\mathcal{S} = \{S \subset V \mid |\delta(S)| \leq p+q-1\}$, Bansal et al.~\cite{bansal2023improved} obtained a primal-dual $16$-approximation algorithm for the {\em Augmenting Small Cuts} problem, which generalizes \globalcshort{1}{q}.  
Recently, the factor has been reduced to $10$~\cite{nutov2024improved} and $5$~\cite{DBLP:journals/corr/abs-2308-15714} via refined analysis.

\begin{theorem}[follows from \cite{bansal2023improved, DBLP:journals/corr/abs-2308-15714}]
    There is a polynomial-time $5$-approximation algorithm for \globalc{1}{q}.
\end{theorem}

Finally, we consider \stc{1}{q}. 
We show that this problem is equivalent to the undirected {\em Minimum Shared Edge} problem: 
%Therein, 
We are given a graph with edge weights and two specified vertices $s,t$.
The task is to find $k$ $s$-$t$ paths with the minimum total weight of shared edges.
Here, an edge is shared if it is contained in at least $2$ paths.  

\begin{restatable}{proposition}{pathequivalence}\label{equal_to_finding_path}
 An edge set $X$ is a feasible solution to \stcshort{1}{q} if and only if there are $(q+1)$ $s$-$t$-paths such that any edge shared by at least two paths belongs to $X$.   
\end{restatable}
\begin{proof}
    To show necessity, we construct a %capacitated network 
    graph $G=(V, E)$ with a capacity function $u$ on the edges, where the capacity $u(e)$ of any edge $e$ is $q+1$ if $e \in X$, and $1$ otherwise. 
    Since $X$ is a feasible solution, by \Cref{prop:cut formulation} the capacity of any $s$-$t$ cut is at least $q+1$. 
    Thus, there exist $q+1$ $s$-$t$-paths such that edges shared by at least two paths 
    %(whose capacity is at least $2$) 
    belong to $X$.
 
    As for the sufficiency, suppose we have $q+1$ $s$-$t$-paths that only share edges in $X$. 
    We claim that the shared edges form a feasible solution of \stcshort{1}{q}. 
    For each cut of size at most $q$, at least one edge must be shared by two paths and this edge is in $X$. 
    Thus, every cut $\delta(S)$ with $|\delta(S)| \leq q$ satisfies $\delta(S) \cap X \neq \emptyset$. 
    By \Cref{prop:cut formulation}, $X$ is a feasible solution.
\end{proof} 

\begin{lemma}\label{lem:min-shared-edges}
    An edge set $X$ is an inclusion-wise minimal solution to \stc{1}{q} if and only if there are $(q+1)$ $s$-$t$-paths such that the shared edges are exactly~$X$.
\end{lemma}

We conclude that \stcshort{1}{q} is equivalent to the Minimum Shared Edge Problem.
%Therefore, by 
Hence, \Cref{lem:min-shared-edges} and the results of~\cite{DBLP:journals/jcss/FluschnikKNS19, DBLP:journals/algorithmica/AssadiENYZ14, DBLP:journals/jco/OmranSZ13}
imply the following.

\begin{theorem}
    \label{thm:stc-implications}
    When parameterized by $q$, \stcshort{1}{q} admits an FPT algorithm for undirected graphs and an XP algorithm for directed graphs.
    Furthermore, \stcshort{1}{q} on directed graphs admits no $\bigO(2^{\log^{1-\epsilon}\max\{q,n\}})$-approximation, unless $\NP \subseteq \DTIME(n^{polylog(n)})$.
\end{theorem}

\subsection{Extension for larger $p$}

Before presenting algorithms for more general cases, we argue that %we can not do much for 
\steinercshort{p}{q} is quite hopeless when both $p$ and $q$ are part of the input.
Indeed, if this is the case, there is no polynomial-time algorithm that verifies feasibility of any given solution unless \P$=$\NP.

By \Cref{prop:cut formulation}, a given protected edge set $X$ is infeasible if and only if there is a terminal-separating cut $\delta(S)$ such that $|\delta(S)| \leq p+q-1$ and $|\delta(S)\cap X| \leq p-1$.
We define and study the complexity of the following  ($A,B$)-bicriteria $s$-$t$-cut problem:
Given an undirected graph with two specified vertices $s,t$ and a subset of edges protected,
decide whether there is an $s$-$t$-cut such that the number of protected edges in the cut is at most $A$ ($=p-1$) and the total number of edges in the cut is at most $B$ ($=p+q-1$) and.
Recall that in the \fndst{p}{q} problem, verifying the feasibility of a solution can be formulated as follows. 
Given an undirected graph with safe and unsafe edges, decide whether there are $p$ edge-disjoint paths between $s$ and $t$ after at most $q$ failures of unsafe edges. 
Hence verifying the feasibility is %also 
equivalent to the ($A,B$)-bicriteria $s$-$t$-cut problem.
We show that the ($A,B$)-bicriteria $s$-$t$-cut problem is \NP-complete, which implies that there is no polynomial-time approximation algorithm for \stc{p}{q} or \fndst{p}{q}, unless $\P=\NP$.

\begin{figure}[t]
\centering
\begin{subfloat}[A $3$-Clique instance, $d=3$.] {

\begin{tikzpicture}[scale=0.55]
    \node[draw=black,circle, minimum size = 0.4cm, inner sep=1pt] (v1) at (0,0){$v_1$};
    \node[draw=black,circle, minimum size = 0.4cm, inner sep=1pt] (v2) at (0,-4){$v_2$};
    \node[draw=black,circle, minimum size = 0.4cm, inner sep=1pt] (v3) at (2,-2){$v_3$};
    \node[draw=black,circle, minimum size = 0.4cm, inner sep=1pt] (v5) at (4,0){$v_5$};
    \node[draw=black,circle, minimum size = 0.4cm, inner sep=1pt] (v6) at (4,-4){$v_6$};
    \node[draw=black,circle, minimum size = 0.4cm, inner sep=1pt] (v4) at (6,-2){$v_4$};

    \draw[]
        (v1) -- node[left] {$e_1$} (v2)
        (v1) -- node[above] {$e_2$} (v5)
        (v1) -- node[right] {$e_3$} (v3)
        (v3) -- node[left,yshift=-6,xshift=-5] {$e_4$} (v4)
        (v4) -- node[right] {$e_6$} (v6)
        (v5) -- node[above,xshift=7,yshift=3] {$e_8$} (v6)
        (v2) -- node[below] {$e_9$} (v6)
        (v5) -- node[right]{$e_5$} (v4)
        (v2) -- node[right] {$e_7$} (v3)
        ;
    \end{tikzpicture}
}
\end{subfloat}
\hspace{10pt} 
\begin{subfloat}[The constructed $(A,B)$-bicriteria $s$-$t$-cut instance. ] {
\centering
\begin{tikzpicture}[scale=0.9]
\foreach \x in {1,...,9}
    \node[draw=black,circle,minimum size = 0.4cm, inner sep=1pt] (e\x) at (\x*1,0){$e_{\x}$};   
\node[] () at (10,0){$V_{\rm edge}$};
\foreach \x in {1,...,6}
    \node[draw=black,circle,minimum size = 0.4cm, inner sep=1pt] (v\x) at (\x*1+1.5,-1.5){$v_{\x}$};
\node[] () at (9,-1.5){$V_{\rm vertex}$};

\draw[]
    (e1.south) -- (v1.north)
    (e1.south) -- (v2.north)
    (e2.south) -- (v1.north)
    (e2.south) -- (v5.north)
    (e3.south) -- (v1.north)
    (e3.south) -- (v3.north)
    (e4.south) -- (v3.north)
    (e4.south) -- (v4.north)
    (e5.south) -- (v4.north)
    (e5.south) -- (v5.north)
    (e6.south) -- (v4.north)
    (e6.south) -- (v6.north)
    (e7.south) -- (v2.north)
    (e7.south) -- (v3.north)
    (e8.south) -- (v5.north)
    (e8.south) -- (v6.north)
    (e9.south) -- (v4.north)
    (e9.south) -- (v6.north)
;

\node[draw=black,circle, minimum size = 0.4cm, inner sep=1pt, below of=v1,yshift=0cm] (t) {$t$}; 
\foreach \x in {1,...,6}
    \draw[blue]
        (t.north) -- node[] { } (v\x.south);
\foreach \x in {0,...,2}
{
    \node[draw=black,circle] (av\x) at (\x*1+7,-1,5){};
    \foreach \y in {1,...,6}
    \draw[] (av\x.north) -- node[] { } (v\y.south);
}
\draw[decoration={brace,mirror,raise=5pt},decorate]
  (5,-3.1) -- node[below=6pt] {$d+1$} (7,-3.1);
\node[draw=black,circle,minimum size = 0.4cm, inner sep=1pt,below of=av1, yshift=-0.3cm] (s) {$s$};
\draw [draw=black,dotted] (7.5,-4.7) rectangle (4.6,-2.7);
\node [text width=2.3cm] () at (9,-3.8) {$C$: clique of size $(d+1)n$};
\end{tikzpicture}
}
\end{subfloat}
\caption{The reduction: the blue edges are protected edges and the black edges are unprotected.}
\label{fig: reduction}
\end{figure}
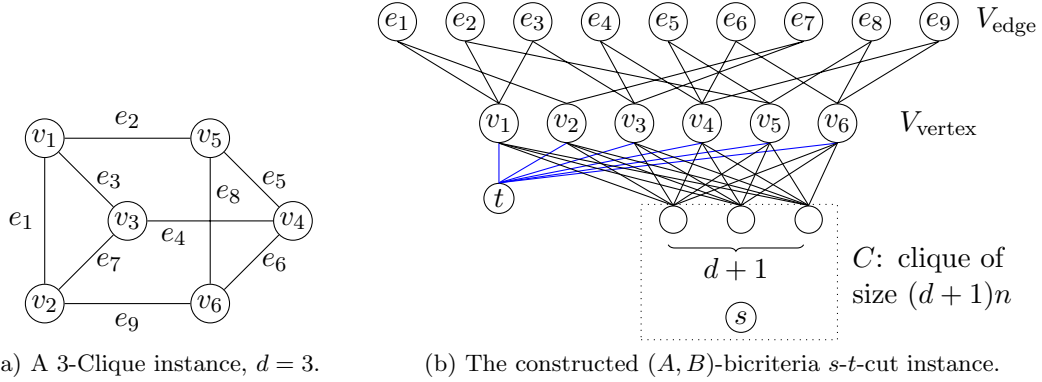
\stcuthard*

\begin{proof}
We use a reduction similar to~\cite{DBLP:conf/mfcs/FominGK13} from $k$-Clique on $d$-regular graphs, which is \NP-complete~\cite{DBLP:books/fm/GareyJ79}. See Figure \ref{fig: reduction} for an illustration.
In $k$-Clique, we are given an undirected graph and we need to decide whether there is a clique of size $k$.
Given an instance of $k$-clique with graph $G=(V,E)$, where $G$ is $d$-regular, we construct an instance of the $(A,B)$-bicriteria $s$-$t$-cut problem $(G',s,t,A,B)$ as follows.
Let $n:=|V|, m:=|E|$.
We create $n$ vertices $V_{\rm vertex}$ corresponding to $V$ and $m$ vertices $V_{\rm edge}$ corresponding to $E$.
In the following, when we say we connect two vertices, then they are connected by an \emph{unprotected} edge by default.
For each $e=(u,v) \in E$, we connect its corresponding vertex to the two vertices corresponding to $u$ and $v$.
Then we create a vertex $t$ and connect it to each vertex in $V_{\rm vertex}$ with \emph{protected} edges.
Create an auxiliary clique $C$ of size $n(d+1)$ and fix an arbitrary vertex in the clique as $s$.
Fix $d+1$ vertices in the clique other than $s$ and fully connect them to each vertex in $V_{\rm vertex}$, which results in a complete bipartite subgraph $K_{n,d+1}$.
Let $A:=k, B:= (d+1)n-k(k-1)$.
We claim that $G$ has a clique of size $k$ if and only if the protected edges form an \emph{infeasible} solution to the instance of the $(A,B)$-bicriteria $s$-$t$-cut problem.

For the first direction, suppose $G$ has a clique $CL=(V_{\rm CL},E_{\rm CL})$ of size $k$. 
Let $S$ include all the vertices in $V_{\rm vertex}$ and $V_{\rm edge}$ corresponding to $V_{CL},E_{CL}$, and the auxiliary clique $C$. 
We show $\delta(S)$ defines an $(A,B)$-bicriteria $s$-$t$-cut. 
In $\delta(S)$, the only protected edges are those edges between $t$ and the vertices corresponding to $V_{CL}$. 
Hence there are exactly $k=A$ protected edges.
As for $|\delta(S)|$, consider the edges between $V_{\rm vertex}$ and $t \cup C$. 
Each vertex in $V_{\rm vertex} \setminus S$ contributes $d+1$ and each vertex in $V_{\rm vertex} \cap S$ contributes $1$.
Now consider the edges between $V_{\rm edge}$ and $V_{\rm vertex}$. 
There are $d \cdot k$ edges incident to $V_{\rm vertex} \cap S$, among which $k(k-1)$ do not contribute to $|\delta(S)|$.
Hence, $|\delta(S)|=(d+1)(n-k)+k+dk-k(k-1)=(d+1)n-k(k-1)= B$ and $\delta(S)$ is an $(A,B)$-bicriteria $s$-$t$-cut. 

For the other direction, suppose there is an $(A,B)$-bicriteria $s$-$t$-cut $\delta(S)$ with $|\delta(S)| \leq B$ such that $\delta(S)$ contains at most $A$ protected edges.
We show that $S \cap V_{\rm vertex}$ corresponds to the vertices of a clique of size $k$ in $G$ and $S \cap V_{\rm edge}$ contains the edges of this clique.
Observe that $|S \cap V_{\rm vertex}| \leq k$ since there are at most $A=k$ protected edges in $\delta(S)$.
We will show that $|S \cap V_{\rm vertex}| = k$ and $|S \cap V_{\rm edge}| \geq k(k-1)/2$.
Furthermore, these $k(k-1)/2$ vertices in $S \cap V_{\rm edge}$
have both their neighbors in $S \cap V_{\rm vertex}$.
Hence $S \cap V_{\rm vertex}$ defines a clique of size $k$ in $G$.

Observe that $S$ must include the whole auxiliary clique C, otherwise $|\delta(S)|$ would exceed $B$.
Let $S'=S \setminus V_{\rm edge}$ and note that $C \subseteq S'$.
We prove that $|\delta(S')|=(d+1)n > B$ by considering the following process.
Starting from $Y:=C$, we add the vertices in $S'\setminus C$ one by one to $Y$.
During the process, $|\delta(Y)|$ does not change since each vertex in $S'\setminus C$ is connected to exactly $d+1$ vertices in $C$, $d$ vertices in $V_{\rm edge}$ and $t$.
Hence $|\delta(S')|=|\delta(C)|=(d+1)n$.
Now starting from $Y=S'$, we add the vertices in $S\setminus S'$ one by one to $Y$.
During the process, the only case that adding a vertex decreases $|\delta(Y)|$ (by $2$) is when both its neighbors are in $S \cap V_{\rm vertex}$. 
Therefore, we have at least $k(k-1)/2$ vertices in $S \setminus S'$, each having both their neighbors in $S \cap V_{\rm vertex}$, since $|\delta(S')|-|\delta(S)| \geq (d+1)n-B=k(k-1)$.
Hence, $|S \cap V_{\rm vertex}| \geq k$, and with the above inequality of $|S \cap V_{\rm vertex}| \leq k$ we have $|S \cap V_{\rm vertex}| = k$.
Further, for any two vertices in $S \cap V_{\rm vertex}$, there is some vertex in $S\cap V_{\rm edge}$ connected to both of them, implying that $S \cap V_{\rm vertex}$ corresponds to the vertices of a clique in $G$.
Hence, $G$ contains a clique of size~$k$.
\end{proof}

On the positive side, if $q$ is a constant, we can enumerate those edge sets $F$ with $|F|\leq q$ such that some terminal pair in $(V,E\setminus F)$ is not $p$-edge-connected. For each of those sets, we need to protect at least one edge in the set, which reduces to the hitting set problem and admits a $q$-approximation where $q$ is the largest size of the sets to be hit~\cite{DBLP:journals/jal/Bar-YehudaE81,DBLP:journals/siamcomp/Hochbaum82}.

In the following, we extend algorithm for \steinercshort{1}{q} to \steinercshort{p}{q} with $p$ being a constant.
The idea is to start from an empty solution and augment the current solution by iteratively increasing the number of protected edges in each critical cut. 
Our algorithm consists of $p$ phases. 
In phase $i$, we are given a partial solution $X_{i-1}$ satisfying that each critical cut contains at least $i-1$ edges in~$X_{i-1}$. 
We then (approximately) solve the following augmentation problem $\mathcal{P}_i$: Add to $X_{i-1}$ a minimum-cost set of edges $Y_i \subseteq E \setminus X_{i-1}$ such that $X_i \coloneqq X_{i-1} \cup Y_i$ includes at least $i$ edges of each critical cut. 
That is, find a set $Y_i$ that includes at least one edge from each critical cut with exactly $i-1$ protected edges in $X_{i-1}$. 
The augmentation problem is solved similarly to the primal-dual framework for \steinercshort{1}{q}. 

Formally, let $\mathcal{S}_0 = \mathcal{S}, X_0 = \emptyset$. 
In phase $i$ with $1 \leq i \leq p$, we define $\mathcal{S}_i = \{S \in \mathcal{S} \mid |\delta(S) \cap X_{i-1}| = i-1\}$, i.e., the critical cuts with exactly $i-1$ protected edges. 
Next, we solve the following problem $\mathcal{P}_i$: find a minimum-cost edge set $Y_i \subseteq E \setminus X_{i-1}$ such that $Y_i \cap \delta(S) \neq \emptyset$ for any $S \in \mathcal{S}_i$. 
Then we set $X_i := X_{i-1} \cup Y_i$ and go on to the next phase. 
To solve $\mathcal{P}_i$, we use a primal-dual algorithm based on the following LP to compute a $(p+q-1)$-approximation solution to $\mathcal{P}_i$ which is essentially the same as \steinerc{1}{q}. 
The approximation ratio is bounded by $p+q-1$ as the size of a critical cut is at most $p+q-1$.

\begin{alignat*}{3}
    \text{min} \quad &\sum_{e \in E\setminus X_{i-1}} c_e x_e & & \quad \quad  \text{max} \quad \sum_{S \in \mathcal{S}_i} y_S & \\
    \text{s.t.} \quad& \sum_{e \in \delta(S) \setminus X_{i-1} } x_e \geq 1   \quad &\forall S \in \mathcal{S}_i  \notag & \quad \quad \text{s.t.} \quad \sum_{S: S \in \mathcal{S}_i, e \in \delta(S) } y_S \leq c_e   \quad &\forall e \in E \setminus X_{i-1}  \notag \\
    & x_e \geq 0 & \hspace{-1cm} \forall e \in E\setminus X_{i-1} & \quad \quad \quad \quad \quad y_S \geq 0 & \hspace{-1cm} \forall S \in \mathcal{S}_i \notag \\
\end{alignat*}

 The only difference is the process of finding a violating set $S \in \mathcal{S}_i$ with respect to some partial solution $X$. However, finding such a violating set is non-trivial. We are only aware of a solution when $p$ is a constant, which we present in the following lemma.

\begin{restatable}{lemma}{constantpfindingcut}
\label{lemma:find_violating_set}
    Given an edge set $X \supseteq X_{i-1}$, there is a polynomial-time algorithm that computes a set $S \in \mathcal{S}_i$ such that $\delta(S) \cap X =\emptyset$ when $p$ is a constant.    
\end{restatable}
\begin{proof}
    Since $X \supseteq X_{i-1}$, we have $|\delta(S) \cap X|  \geq i-1$ for any $S \in \mathcal{S}$. 
    It suffices to find some $S \in \mathcal{S}$ with $|\delta(S) \cap X| = i-1 \leq p$. 
    To this end, we guess the edge set $X'=\delta(S) \cap X$. Note that $|X'|< p$. 
    Further, for each edge $e=(u,v)$ in $X'$, we guess whether $u \in S, v \notin S$ or $u \notin S, v \in S$. 
    Thus the number of possibilities is at most $\binom{m}{p}\cdot 2^p$, which is polynomial when $p$ is constant.
    For the edges in $X'$, let the set of endpoints in $S$ be $A$, and the other endpoints be $B$. 
    It reduces to finding some $S \in \mathcal{S}$ with $A \subseteq S$, $B \notin S$ and $\delta(S) \cap X =X'$. 
    This can be achieved by identifying the vertices in $A$ and $B$ by a new vertex $v_A$ and $v_B$, respectively, contracting edges in $X \setminus X'$, and computing a minimum $v_A$-$v_B$ cut in the resulting graph.
    If the cut has size less than $p + q$, then this cut belongs to $\mathcal{S}_i$ and has $i-1$ edges in $X$.
\end{proof}
To bound the total cost of the $p$ phases of our algorithm, we use the following LP relaxation and its dual for the analysis. 
The constraints $x_e \leq 1$ cannot be omitted as we do for $p=1$. 
Otherwise, an edge may be "protected" multiple times, which is not allowed.

\begin{minipage}{0.37\textwidth}
\begin{alignat*}{3}
    \text{min} \quad &\sum_{e \in E} c_e x_e  \\ 
    \text{s.t.} \quad& \sum_{e \in \delta(S)} x_e \geq p  \quad &\forall S \in \mathcal{S}  \notag \\
    & 0 \leq x_e \leq 1 & \forall e \in E \notag \\
\end{alignat*}
\end{minipage}
\hfill
\begin{minipage}{0.63\textwidth}
\begin{alignat*}{3}
    \text{max} \quad &\sum_{S \in \mathcal{S}} p\cdot y_S-\sum_{e \in E}z_e  \\ 
    \text{s.t.} \quad& \sum_{S: S \in \mathcal{S},e \in \delta(S)} y_S -z_e \leq c_e  \quad &\forall e \in E  \notag \\
    & y_S,z_e \geq 0 & \forall S \in \mathcal{S}, \forall e \in E \notag \\
\end{alignat*}
\end{minipage}

In the following lemma, we compare the optimal cost of the augmentation problem $\mathcal{P}_i$ to the optimal cost of \steinerc{p}{q}.
\begin{lemma}\label{lemma: dual mapping}
    Given a feasible dual solution $y^{(i)}$ of $\mathcal{P}_i$, we can construct a feasible dual solution $y$ of \steinerc{p}{q} such that $$\sum_{S \in \mathcal{S}_i} y^{(i)}_S \leq \frac{1}{p-i+1} \Big( \sum_{S \in \mathcal{S}} p\cdot y_S - \sum_{e \in E}z_e \Big) \ .$$
\end{lemma}

\begin{proof}
    Let $y_S = y^{(i)}_S$ for $S \in \mathcal{S}_i$ and $y_S = 0$ for $S \in \mathcal{S}\setminus \mathcal{S}_i$. Let $z_e = 0$ for $e \in E\setminus X_{i-1}$ and $z_e = \sum_{S: S \in \mathcal{S}_i,e \in \delta(S)} y^{(i)}_S$ for $e \in X_{i-1}$. We claim that $(y,z)$ forms a feasible dual solution. For any $e \in X_{i-1}$, $\sum_{S: S \in \mathcal{S},e \in \delta(S)} y_S -z_e = 0$ by definition. For $e \in E \setminus X_{i-1}$, $\sum_{S: S \in \mathcal{S},e \in \delta(S)} y_S -z_e = \sum_{S: S \in \mathcal{S},e \in \delta(S)} y_S^{(i)} \leq c_e$. Next, we compare the dual objective values of $y$ and $y^{(i)}$. We have
    \[
    \sum_{S \in \mathcal{S}} p\cdot y_S - \sum_{e \in E}z_e =   \sum_{S \in \mathcal{S}_i} p\cdot y^{(i)}_S - \sum_{e \in X_{i-1}} \sum_{S: S \in \mathcal{S}_i,e \in \delta(S)} y^{(i)}_S =   \sum_{S \in \mathcal{S}_i} y^{(i)}_S (p-|\delta(S) \cap X_{i-1}|). 
    \]
    %Note that 
    By definition of $\mathcal{S}_i$, for any $S \in \mathcal{S}_i$, we have $|\delta(S) \cap X_{i-1}| = i-1$. Thus, we conclude:
    \[
    \sum_{S \in \mathcal{S}} p\cdot y_S - \sum_{e \in E}z_e = (p-i+1) \sum_{S \in \mathcal{S}_i} y^{(i)}_S. \qedhere 
    \]
    %which concludes the proof.
\end{proof}

\thmapprox*
\begin{proof}
    The algorithm consists of $p$ phases. In phase $i$, we apply \Cref{lemma:find_violating_set} and the primal-dual framework for \steinercshort{1}{q} to find a $(p+q-1)$-approximation solution for the augmentation problem $\mathcal{P}_i$. 
    By \Cref{lemma: dual mapping}, the cost of $Y_i$ in phase $i$ is at most $(p+q-1) \cdot \opt(\mathcal{P}_i) \leq \frac{p+q-1}{p-i+1} \opt$.
    Thus the total cost is at most $ \sum_{i=1}^{p} c(Y_i) \leq \sum_{i=1}^{p} \frac{p+q-1}{p-i+1} \opt \leq H_p \cdot (p+q-1) \cdot \opt$, where $H_p$ is the $p$-th harmonic number. 
    Using $H_p \leq \log(p) + 1$, we obtain the theorem.
\end{proof}

For \globalc{p}{q}, we can approximately solve the augmentation problem without requiring $p$ to be a constant.
Indeed, we reduce finding the critical cuts to finding certain $2$-approximate minimum cuts in $G$, where each edge $e$ has a capacity of $\frac{p+q}{i-1}$ if $e \in X_{i-1}$ and $1$ otherwise.
These cuts can be enumerated in polynomial time~\cite{DBLP:conf/soda/Karger93, DBLP:journals/siamdm/NagamochiNI97}.

{\bf Algorithm 4.\ }
In phase $1$, we apply the $5$-approximation algorithm from~\cite{DBLP:journals/corr/abs-2308-15714}. 
That is, we compute~$X_1$ such that for any $S \in \mathcal{S}_0 = \mathcal{S}$, $X_1 \cap \delta(S) \neq \emptyset$. 
For phase $i$ with $2 \leq i \leq p$, we approximately solve the augmentation problem $\mathcal{P}_i$ by reducing it to Set Cover. 
Here, we view a set $S \in \mathcal{S}_i$ as an element in the Set Cover instance and view an edge $e \in E \setminus X_{i-1}$ as a set in the Set Cover instance.
We use either the $\bigO(\log N)$-approximation~\cite{DBLP:journals/mor/Chvatal79} where $N$ is the number of elements to be covered, or the $f$-approximation~\cite{DBLP:journals/jal/Bar-YehudaE81,DBLP:journals/siamcomp/Hochbaum82} where $f$ is the maximum number of sets in which an element is contained.
Note that applying \Cref{lemma: dual mapping} requires a dual feasible solution, which is fortunately a byproduct of these Set Cover algorithms.

\thmapproxtwo*
\begin{proof}

The cost of phase $1$ is no more than $5 \cdot \opt$. 
For phase $i$ with $2 \leq i \leq p$, we apply Set Cover algorithms explicitly.
We show that the number of elements to be covered is $|\mathcal{S}_i| = \bigO(|V|^4)$ and we can construct the Set Cover instance in polynomial time.
To this end, we assign different capacities to edges in $X_{i-1}$ and other edges such that for any $S \in \mathcal{S}_i$, $\delta(S)$ is a $2$-approximate minimum cut with respect to the capacity function. 
By Karger’s bound~\cite{DBLP:conf/soda/Karger93}, the number of $2$-approximate minimum cuts is $\bigO(|V|^4)$ and we can enumerate them in polynomial time~\cite{DBLP:journals/siamdm/NagamochiNI97}. 
Formally, let the capacity of edges in $X_{i-1}$ be $\frac{p+q}{i-1}$ and the capacity of edges in $E\setminus X_{i-1}$ be 1. 
Given any cut $C$, the capacity of $C$ is at least $p+q$ since it either contains at least $p+q$ edges or contains at least $i-1$ edges in $X_{i-1}$. 
For any $S \in \mathcal{S}_i$, the capacity of $\delta(S)$ is at most $(i-1)\frac{p+q}{i-1} + p+q-1 < 2(p+q)$. Thus $\delta(S)$ defines a $2$-approximate minimum cut and we can find all the sets in $\mathcal{S}_i$ in polynomial time.

Further, in the constructed Set Cover instance, an element is contained in at most $p+q-1$ sets since $|\delta(S)| \leq p+q-1$ for any $S \in \mathcal{S}_i$. 
Thus, we can compute an augmenting edge set $X_i \setminus X_{i-1}$ whose cost is $ \bigO(\min\{\log n, p+q\} \cdot \sum_{S \in \mathcal{S}_i}y^{(i)}_S)$ where $y^{(i)}$ is the dual feasible solution of $\mathcal{P}_i$. Combining it with \Cref{lemma: dual mapping}, we conclude that the algorithm is an $\bigO(\log p \cdot \min \{\log n, p+q\})$-approximation.
\end{proof}

\section{Conclusion}
We examine Connectivity Preservation from two perspectives. For small values of $p$ and $q$, we focus on polynomial-time exact algorithms. For large values of $p$ and $q$, we show hardness and devise approximation algorithms. Nonetheless, there remain some gaps between cases solvable in polynomial time and NP-hard ones. 
In particular, it remains open whether \globalcshort{1}{q} admits any polynomial-time exact algorithm for constant $q \geq 3$. 
Another interesting problem is \globalcshort{1}{q} with $q$ being the capacity of the minimum cuts, i.e., finding a minimum-cost edge set that intersects with all the minimum cuts. 
Note that for the $s$-$t$-connectivity variant, this can be tackled via Min-cost Flow techniques.

%%
%% Bibliography
%%

%% Please use bibtex, 

\bibliography{ref}

\begin{thebibliography}{10}

\bibitem{abbas2017improving}
Waseem Abbas, Aron Laszka, Yevgeniy Vorobeychik, and Xenofon Koutsoukos.
\newblock Improving network connectivity using trusted nodes and edges.
\newblock In {\em 2017 American Control Conference (ACC)}, pages 328--333. IEEE, 2017.

\bibitem{DBLP:journals/mp/AdjiashviliHM22}
David Adjiashvili, Felix Hommelsheim, and Moritz M{\"{u}}hlenthaler.
\newblock Flexible graph connectivity.
\newblock {\em Math. Program.}, 192(1):409--441, 2022.

\bibitem{DBLP:conf/swat/AdjiashviliHMS22}
David Adjiashvili, Felix Hommelsheim, Moritz M{\"{u}}hlenthaler, and Oliver Schaudt.
\newblock Fault-tolerant edge-disjoint s-t paths - beyond uniform faults.
\newblock In {\em {SWAT}}, volume 227 of {\em LIPIcs}, pages 5:1--5:19. Schloss Dagstuhl - Leibniz-Zentrum f{\"{u}}r Informatik, 2022.

\bibitem{DBLP:journals/mp/AdjiashviliSZ15}
David Adjiashvili, Sebastian Stiller, and Rico Zenklusen.
\newblock Bulk-robust combinatorial optimization.
\newblock {\em Math. Program.}, 149(1-2):361--390, 2015.

\bibitem{DBLP:journals/algorithmica/AssadiENYZ14}
Sepehr Assadi, Ehsan Emamjomeh{-}Zadeh, Ashkan Norouzi{-}Fard, Sadra Yazdanbod, and Hamid Zarrabi{-}Zadeh.
\newblock The minimum vulnerability problem.
\newblock {\em Algorithmica}, 70(4):718--731, 2014.

\bibitem{DBLP:journals/corr/abs-2308-15714}
Ishan Bansal.
\newblock A global analysis of the primal-dual method for pliable families, 2024.
\newblock URL: \url{https://arxiv.org/abs/2308.15714}, \href {https://arxiv.org/abs/2308.15714} {\path{arXiv:2308.15714}}.

\bibitem{DBLP:conf/approx/BansalCGI23}
Ishan Bansal, Joe Cheriyan, Logan Grout, and Sharat Ibrahimpur.
\newblock Algorithms for 2-connected network design and flexible steiner trees with a constant number of terminals.
\newblock In {\em {APPROX/RANDOM}}, volume 275 of {\em LIPIcs}, pages 14:1--14:14. Schloss Dagstuhl - Leibniz-Zentrum f{\"{u}}r Informatik, 2023.

\bibitem{bansal2023improved}
Ishan Bansal, Joseph Cheriyan, Logan Grout, and Sharat Ibrahimpur.
\newblock Improved approximation algorithms by generalizing the primal-dual method beyond uncrossable functions.
\newblock {\em Algorithmica}, 86(8):2575--2604, 2024.

\bibitem{bansal2024improvedapproximationalgorithmsflexible}
Ishan Bansal, Joseph Cheriyan, Sanjeev Khanna, and Miles Simmons.
\newblock Improved approximation algorithms for flexible graph connectivity and capacitated network design, 2024.
\newblock URL: \url{https://arxiv.org/abs/2411.18809}, \href {https://arxiv.org/abs/2411.18809} {\path{arXiv:2411.18809}}.

\bibitem{DBLP:journals/jal/Bar-YehudaE81}
Reuven Bar{-}Yehuda and Shimon Even.
\newblock A linear-time approximation algorithm for the weighted vertex cover problem.
\newblock {\em J. Algorithms}, 2(2):198--203, 1981.

\bibitem{DBLP:journals/ipl/BernP89}
Marshall~W. Bern and Paul~E. Plassmann.
\newblock The steiner problem with edge lengths 1 and 2.
\newblock {\em Inf. Process. Lett.}, 32(4):171--176, 1989.

\bibitem{DBLP:journals/siamcomp/BienstockD93}
Daniel Bienstock and Nicole Diaz.
\newblock Blocking small cuts in a network, and related problems.
\newblock {\em {SIAM} J. Comput.}, 22(3):482--499, 1993.

\bibitem{DBLP:journals/mp/BoydCHI24}
Sylvia~C. Boyd, Joseph Cheriyan, Arash Haddadan, and Sharat Ibrahimpur.
\newblock Approximation algorithms for flexible graph connectivity.
\newblock {\em Math. Program.}, 204(1):493--516, 2024.

\bibitem{DBLP:conf/stoc/CecchettoTZ21}
Federica Cecchetto, Vera Traub, and Rico Zenklusen.
\newblock Bridging the gap between tree and connectivity augmentation: unified and stronger approaches.
\newblock In {\em {STOC}}, pages 370--383. {ACM}, 2021.

\bibitem{DBLP:conf/stoc/CenLP22}
Ruoxu Cen, Jason Li, and Debmalya Panigrahi.
\newblock Edge connectivity augmentation in near-linear time.
\newblock In {\em {STOC}}, pages 137--150. {ACM}, 2022.

\bibitem{DBLP:journals/algorithmica/ChakrabartyCKK15}
Deeparnab Chakrabarty, Chandra Chekuri, Sanjeev Khanna, and Nitish Korula.
\newblock Approximability of capacitated network design.
\newblock {\em Algorithmica}, 72(2):493--514, 2015.

\bibitem{chekuri2023approximation}
Chandra Chekuri and Rhea Jain.
\newblock Approximation algorithms for network design in non-uniform fault models.
\newblock In {\em 50th International Colloquium on Automata, Languages, and Programming (ICALP 2023)}. Schloss Dagstuhl-Leibniz-Zentrum f{\"u}r Informatik, 2023.

\bibitem{chekuri2024approximation}
Chandra Chekuri and Rhea Jain.
\newblock Approximation algorithms for network design in non-uniform fault models, 2024.
\newblock URL: \url{https://arxiv.org/abs/2403.15547}, \href {https://arxiv.org/abs/2403.15547} {\path{arXiv:2403.15547}}.

\bibitem{DBLP:journals/mor/Chvatal79}
Vasek Chv{\'{a}}tal.
\newblock A greedy heuristic for the set-covering problem.
\newblock {\em Math. Oper. Res.}, 4(3):233--235, 1979.

\bibitem{dinitz1973structure}
Efim~A. Dinits, Alexander~V. Karzanov, and Micael~V. Lomonosov.
\newblock On the structure of a family of minimum weighted cuts in a graph.
\newblock {\em Studies in Discrete Optimization}, pages 209--306, 1976.

\bibitem{duarte2021computing}
Gabriel~L Duarte, Hiroshi Eto, Tesshu Hanaka, Yasuaki Kobayashi, Yusuke Kobayashi, Daniel Lokshtanov, Lehilton~LC Pedrosa, Rafael~CS Schouery, and U{\'e}verton~S Souza.
\newblock Computing the largest bond and the maximum connected cut of a graph.
\newblock {\em Algorithmica}, 83:1421--1458, 2021.

\bibitem{DBLP:journals/siamcomp/EswaranT76}
Kapali~P. Eswaran and Robert~Endre Tarjan.
\newblock Augmentation problems.
\newblock {\em {SIAM} J. Comput.}, 5(4):653--665, 1976.

\bibitem{DBLP:journals/jcss/FluschnikKNS19}
Till Fluschnik, Stefan Kratsch, Rolf Niedermeier, and Manuel Sorge.
\newblock The parameterized complexity of the minimum shared edges problem.
\newblock {\em J. Comput. Syst. Sci.}, 106:23--48, 2019.

\bibitem{DBLP:conf/mfcs/FominGK13}
Fedor~V. Fomin, Petr~A. Golovach, and Janne~H. Korhonen.
\newblock On the parameterized complexity of cutting a few vertices from a graph.
\newblock In {\em {MFCS}}, volume 8087 of {\em Lecture Notes in Computer Science}, pages 421--432. Springer, 2013.

\bibitem{DBLP:journals/siamcomp/FredericksonJ81}
Greg~N. Frederickson and Joseph~F. J{\'{a}}J{\'{a}}.
\newblock Approximation algorithms for several graph augmentation problems.
\newblock {\em {SIAM} J. Comput.}, 10(2):270--283, 1981.

\bibitem{DBLP:conf/soda/Gabow90}
Harold~N. Gabow.
\newblock Data structures for weighted matching and nearest common ancestors with linking.
\newblock In {\em {SODA}}, pages 434--443. {SIAM}, 1990.

\bibitem{DBLP:books/fm/GareyJ79}
M.~R. Garey and David~S. Johnson.
\newblock {\em Computers and Intractability: {A} Guide to the Theory of NP-Completeness}.
\newblock W. H. Freeman, 1979.

\bibitem{DBLP:journals/algorithmica/GargVY97}
Naveen Garg, Vijay~V. Vazirani, and Mihalis Yannakakis.
\newblock Primal-dual approximation algorithms for integral flow and multicut in trees.
\newblock {\em Algorithmica}, 18(1):3--20, 1997.

\bibitem{DBLP:conf/soda/GoemansGPSTW94}
Michel~X. Goemans, Andrew~V. Goldberg, Serge~A. Plotkin, David~B. Shmoys, {\'{E}}va Tardos, and David~P. Williamson.
\newblock Improved approximation algorithms for network design problems.
\newblock In {\em {SODA}}, pages 223--232. {ACM/SIAM}, 1994.

\bibitem{DBLP:conf/soda/HeHS24a}
Zhongtian He, Shang{-}En Huang, and Thatchaphol Saranurak.
\newblock Cactus representation of minimum cuts: Derandomize and speed up.
\newblock In {\em {SODA}}, pages 1503--1541. {SIAM}, 2024.

\bibitem{DBLP:journals/siamcomp/Hochbaum82}
Dorit~S. Hochbaum.
\newblock Approximation algorithms for the set covering and vertex cover problems.
\newblock {\em {SIAM} J. Comput.}, 11(3):555--556, 1982.

\bibitem{DBLP:journals/combinatorica/Jain01}
Kamal Jain.
\newblock A factor 2 approximation algorithm for the generalized steiner network problem.
\newblock {\em Comb.}, 21(1):39--60, 2001.

\bibitem{DBLP:conf/soda/Karger93}
David~R. Karger.
\newblock Global min-cuts in {RNC}, and other ramifications of a simple min-cut algorithm.
\newblock In {\em {SODA}}, pages 21--30. {ACM/SIAM}, 1993.

\bibitem{DBLP:journals/siamcomp/KortsarzKL04}
Guy Kortsarz, Robert Krauthgamer, and James~R. Lee.
\newblock Hardness of approximation for vertex-connectivity network design problems.
\newblock {\em {SIAM} J. Comput.}, 33(3):704--720, 2004.

\bibitem{DBLP:books/cu/NI2008}
Hiroshi Nagamochi and Toshihide Ibaraki.
\newblock {\em Algorithmic Aspects of Graph Connectivity}, volume 123 of {\em Encyclopedia of Mathematics and its Applications}.
\newblock Cambridge University Press, 2008.

\bibitem{DBLP:journals/siamdm/NagamochiNI97}
Hiroshi Nagamochi, Kazuhiro Nishimura, and Toshihide Ibaraki.
\newblock Computing all small cuts in an undirected network.
\newblock {\em {SIAM} J. Discret. Math.}, 10(3):469--481, 1997.

\bibitem{nutov2024improved}
Zeev Nutov.
\newblock Improved approximation ratio for covering pliable set families, 2024.
\newblock URL: \url{https://arxiv.org/abs/2404.00683}, \href {https://arxiv.org/abs/2404.00683} {\path{arXiv:2404.00683}}.

\bibitem{DBLP:journals/jco/OmranSZ13}
Masoud~T. Omran, J{\"{o}}rg{-}R{\"{u}}diger Sack, and Hamid Zarrabi{-}Zadeh.
\newblock Finding paths with minimum shared edges.
\newblock {\em J. Comb. Optim.}, 26(4):709--722, 2013.

\bibitem{DBLP:conf/waoa/Pritchard10}
David Pritchard.
\newblock \emph{k}-edge-connectivity: Approximation and {LP} relaxation.
\newblock In {\em {WAOA}}, volume 6534 of {\em Lecture Notes in Computer Science}, pages 225--236. Springer, 2010.

\bibitem{DBLP:conf/focs/TraubZ21}
Vera Traub and Rico Zenklusen.
\newblock A better-than-2 approximation for weighted tree augmentation.
\newblock In {\em {FOCS}}, pages 1--12. {IEEE}, 2021.

\bibitem{DBLP:conf/soda/TraubZ22}
Vera Traub and Rico Zenklusen.
\newblock Local search for weighted tree augmentation and steiner tree.
\newblock In {\em {SODA}}, pages 3253--3272. {SIAM}, 2022.

\bibitem{DBLP:conf/stoc/TraubZ23}
Vera Traub and Rico Zenklusen.
\newblock A (1.5+{\(\epsilon\)})-approximation algorithm for weighted connectivity augmentation.
\newblock In {\em {STOC}}, pages 1820--1833. {ACM}, 2023.

\bibitem{DBLP:journals/combinatorica/WilliamsonGMV95}
David~P. Williamson, Michel~X. Goemans, Milena Mihail, and Vijay~V. Vazirani.
\newblock A primal-dual approximation algorithm for generalized steiner network problems.
\newblock {\em Comb.}, 15(3):435--454, 1995.

\bibitem{zhang2017enhancing}
Jianan Zhang, Eytan Modiano, and David Hay.
\newblock Enhancing network robustness via shielding.
\newblock {\em IEEE/ACM Transactions on Networking}, 25(4):2209--2222, 2017.

\end{thebibliography}

%\newpage 
%\appendix

\end{document}